\documentclass[a4paper,UKenglish,cleveref, autoref, thm-restate,authorcolumns]{lipics-v2019}
\usepackage{microtype}%if unwanted, comment out or use option "draft"
\usepackage{amsthm,amsmath}
\usepackage{thmtools}
\usepackage{amssymb}
\usepackage{hyperref}
\usepackage{graphicx}

\usepackage{xspace,cite}

\usepackage{xcolor}

\usepackage[ruled, vlined,boxed,commentsnumbered,linesnumbered]{algorithm2e}
\definecolor{dark-red}{rgb}{0.4,0.15,0.15}
\definecolor{dark-blue}{rgb}{0.15,0.15,0.4}
\definecolor{medium-blue}{rgb}{0,0,0.5}
\definecolor{gray}{rgb}{0.5,0.5,0.5}
\definecolor{color-Ig}{rgb}{0.15,0.7,0.15}

\hypersetup{
    colorlinks, linkcolor={dark-red},
    citecolor={dark-blue}, urlcolor={medium-blue}
}

\newcommand{\NP}{\ensuremath{\mathsf{NP}}\xspace}

\newcommand{\Hcal}{\mathcal{H}}

\newcommand{\Lcal}{\mathcal{L}}

\newcommand{\Ocal}{\mathcal{O}}

\newcommand{\Rcal}{\mathcal{R}}

\newcommand{\ETH}{{\sf ETH}\xspace}
\newcommand{\FPT}{{\sf FPT}\xspace}

\newcommand{\probl}[3]{
  \begin{flushleft}
    \fbox{
      \begin{minipage}{.99\textwidth}
        \noindent {\sc #1}\\
        {\bf Input:} #2\\
        {\bf Output:} #3
      \end{minipage}}
  \end{flushleft}
}

\newcommand{\tw}{{\sf tw}\xspace}

\newcommand{\todo}[1][]{%
  \ifx/#1/%
    \textcolor{red}{TODO!}%
  \else%

    \textcolor{red}{todo: #1}%
  \fi%
}

\newcommand{\false}{{\sf false}\xspace}

\theoremstyle{plain}

\hypersetup{
colorlinks = true,
linkcolor = black!30!blue,
citecolor = black!30!green
}

\title{Hitting forbidden induced subgraphs on bounded treewidth graphs}

\titlerunning{Hitting forbidden  induced subgraphs on bounded treewidth graphs} %TODO optional, please use if title is longer than one line

\author{Ignasi Sau}{LIRMM, Universit\'e de Montpellier, CNRS, Montpellier, France}{ignasi.sau@lirmm.fr}{https://orcid.org/0000-0002-8981-9287}{CAPES-PRINT Institutional Internationalization Program, process 88887.371209/ 2019-00, and projects DEMOGRAPH (ANR-16-CE40-0028), ESIGMA (ANR-17-CE23-0010), and ELIT (ANR-20-CE48-0008-01).}

\author{U\'everton dos Santos Souza}{Instituto de Computa\c c\~ao, Universidade Federal Fluminense, Niter\'oi, Brazil}{ueverton@ic.uff.br}{https://orcid.org/0000-0002-5320-9209}{Grant E-26/203.272/2017 Rio de Janeiro Research Foundation (FAPERJ) and Grant 303726/2017-2 National Council for Scientific and Technological Development (CNPq).}

\authorrunning{Ignasi Sau and U\'everton S. Souza} %TODO mandatory. First: Use abbreviated first/middle names. Second (only in severe cases): Use first author plus 'et al.'

\Copyright{Ignasi Sau and U\'everton S. Souza} %TODO mandatory, please use full first names. LIPIcs license is "CC-BY";  http://creativecommons.org/licenses/by/3.0/

%\textcolor{red}{Replace ccsdesc macro with valid one}} %TODO mandatory: Please choose ACM 2012 classifications from https://dl.acm.org/ccs/ccs_flat.cfm


\category{} %optional, e.g. invited paper

\relatedversion{A conference version of this article appeared in the \emph{Proceedings of the 45th International Symposium on Mathematical Foundations of Computer Science (MFCS), volume 170 of LIPIcs, pages 82:1--82:15, \textbf{2020}}. A full version of the paper is permanently available at \url{https://arxiv.org/abs/2004.08324}.}

\supplement{}%optional, e.g. related research data, source code, ... hosted on a repository like zenodo, figshare, GitHub, ...

\acknowledgements{}%I want to thank \dots}%optional

\nolinenumbers %uncomment to disable line numbering

\hideLIPIcs  %uncomment to remove references to LIPIcs series (logo, DOI, ...), e.g. when preparing a pre-final version to be uploaded to arXiv or another public repository

\EventNoEds{2}
\EventLongTitle{}
\EventShortTitle{arXiv preprint}
\EventAcronym{}
\EventYear{2020}
\EventDate{August 24--28, 2020}
\EventLocation{Prague, Czech Republic}
\EventLogo{}
\SeriesVolume{}
\ArticleNo{1}
\begin{document}

\maketitle

\begin{abstract}

For a fixed graph $H$, the \textsc{$H$-IS-Deletion} problem asks, given a graph $G$, for the minimum size of a set $S \subseteq V(G)$ such that $G\setminus S$ does not contain $H$ as an induced subgraph. Motivated by previous work about hitting (topological) minors and subgraphs on bounded treewidth graphs, we are interested in determining, for a fixed graph $H$, the smallest function $f_H(t)$ such that \textsc{$H$-IS-Deletion} can be solved in time $f_H(t) \cdot n^{\Ocal(1)}$ assuming the Exponential Time Hypothesis (\ETH), where $t$ and $n$ denote the treewidth and the number of vertices of the input graph, respectively.

We show that $f_H(t) = 2^{\Ocal(t^{h-2})}$ for every graph $H$ on $h \geq 3$ vertices, and that $f_H(t) = 2^{\Ocal(t)}$ if $H$ is a clique or an independent set. We present a number of lower bounds by generalizing a reduction of Cygan et al.~[Inf. Comput. 2017] for the subgraph version. In particular, we show that when $H$ deviates slightly from a clique, the function $f_H(t)$ suffers a sharp jump: if $H$ is obtained from a clique of size $h$ by removing one edge, then $f_H(t) = 2^{\Theta(t^{h-2})}$. We also show that $f_H(t) = 2^{\Omega(t^{h})}$ when $H=K_{h,h}$, and this reduction answers an open question of Mi. Pilipczuk~[MFCS 2011] about the function  $f_{C_4}(t)$ for the subgraph version.

Motivated by Cygan et al.~[Inf. Comput. 2017], we also consider the colorful variant of the problem, where each vertex of $G$ is colored
with some color from $V(H)$ and we require to hit only induced copies of $H$ with matching colors. In this case, we determine, under the \ETH, the  function $f_H(t)$ for every connected graph $H$ on $h$ vertices: if $h\leq 2$ the problem can be solved in polynomial time; if $h\geq 3$,
$f_H(t) = 2^{\Theta(t)}$ if $H$ is a clique, and $f_H(t) = 2^{\Theta(t^{h-2})}$ otherwise.
%By Courcelle's Theorem~\cite{Courcelle90}, for every graph $H$, \textsc{$H$-IS-Deletion} can be solved in time $f_H(\tw) \cdot n^{\Ocal(1)}$, where $\tw$ and $n$ denote the treewidth and the number of vertices of $G$, respectively. The objective of this article is to determine, for a fixed $H$, which is the asymptotically smallest function $f_H(\tw)$ such that \textsc{$H$-IS-Deletion} can be solved in time $f_H(\tw) \cdot n^{\Ocal(1)}$, subject to the Exponential Time Hypothesis (\ETH) that states, in a simplified version, that the 3-\textsc{Sat} problem on $n$ variables cannot be solved in time $2^{o(n)}$; see~\cite{ImpagliazzoP01-ETH} for more details.
\end{abstract}

\newpage

\section{Introduction}
\label{sec:intro}

Graph modification problems play a central role  in modern algorithmic graph theory. In general, such a problem is determined by a target graph class ${\cal G}$ and some prespecified set ${\cal M}$ of allowed {\sl  local}  modifications, and the question is, given
an input graph $G$ and an integer $k$, whether it is possible to
transform $G$ to a graph in ${\cal G}$ by applying $k$ modification operations from ${\cal M}$.
A wealth of graph problems can be formulated
for different instantiations of ${\cal G}$ and ${\cal M}$, and applications span diverse topics such as computational biology, computer vision, machine learning, networking, and sociology~\cite{FominSM15grap,BodlaenderHL14grap,CrespelleDFG13asur}.

The most studied local modification operation in the literature is vertex deletion and, among the target graph classes, of particular relevance are the ones defined by {\sl excluding} the graphs  in a family ${\cal F}$ according to some natural graph containment relation, such as minor, topological minor, subgraph, or induced subgraph. By the well-known classification result of Lewis and Yannakakis~\cite{LeYa80}, all interesting cases of these problems are \NP-hard.

One of the common strategies to cope with \NP-hard problems is that of parameterized complexity~\cite{DF13,CyganFKLMPPS15}, where the core idea is to identify a parameter $k$ associated with an input of size $n$ that allows for an algorithm in time $f(k) \cdot n^{\Ocal(1)}$, called \emph{fixed-parameter tractable} (or \FPT for short). A natural goal within parameterized algorithms is the quest for the ``best possible'' function $f(k)$ in an \FPT algorithm. Usually, the working hypothesis to prove lower bounds is the \emph{Exponential Time Hypothesis} (\ETH) that states, in a simplified version, that the 3-\textsc{Sat} problem on $n$ variables cannot be solved in time $2^{o(n)}$; see~\cite{ImpagliazzoP01-ETH,ImpagliazzoP01} for more details.

Among graph parameters,  one of the most successful ones is \emph{treewidth}, which --informally speaking-- quantifies the topological resemblance of a graph to a tree. The celebrated theorem due to Courcelle~\cite{Courcelle90} states that every graph problem that can be expressed in Monadic Second Order logic is solvable in time $f(t) \cdot n$ on $n$-vertex graphs given along with a tree decomposition of width at most $t$. In particular, it applies to most vertex deletion problems discussed above. A very active area in parameterized complexity during the last years consists in optimizing, under the \ETH, the function $f(t)$ %given by the theorem of Courcelle Theorem
for several classes of vertex deletion problems. As a byproduct, several cutting-edge techniques for obtaining both lower bounds~\cite{LokshtanovMS11} and algorithms~\cite{BodlaenderCKN15,CyganNPPRW11,FominLPS16} have been obtained, which have become part of the standard toolbox of parameterized complexity. Obtaining tight bounds under the \ETH for this kind of vertex deletion problems is, in general, a challenging task, as we proceed to discuss.

Let $H$ be a fixed graph and let $\prec$ be a fixed graph containment relation. In the \textsc{$H$-$\prec$-Deletion} (meta)problem, given an $n$-vertex graph $G$, the objective is to find a set $S \subseteq V(G)$ of minimum size such that $G \setminus S$ does not contain $H$ according to containment relation $\prec$. We parameterize the problem by the treewidth of $G$, denoted by $t$, and the objective is to find the smallest function $f_H(t)$ such that \textsc{$H$-$\prec$-Deletion} can be solved in time $f_H(t) \cdot n^{\Ocal(1)}$.

%For a fixed finite family of graphs $\Fcal$, the \textsc{$\Fcal$-M-Deletion} problem consists in finding a minimum-sized set of vertices of an input graph to obtain a graph that does not contain any of the graphs in $\Fcal$ as a {\sl minor}. \textsc{$\Fcal$-M-Deletion} parameterized by treewidth

The case $\prec =$ `minor'  has been object of intense study during the last years~\cite{CyganNPPRW11,RueST14,BodlaenderCKN15,DornPBF10,JansenLS14,Pilipczuk17}, culminating in a tight dichotomy about the function $f_H(t)$ when $H$ is connected~\cite{BasteST20-part1,BasteST20-part2,BasteST20-part3,BasteST20}.

The case $\prec =$ `topological minor' has been also studied recently~\cite{BasteST20-part1,BasteST20-part2,BasteST20-part3}, but we are still far from obtaining a complete characterization of the function $f_H(t)$. For both minors and topological minors, so far there is no graph $H$ such that\footnote{For conciseness, we use (in a non-standard way) the asymptotic notations $\Omega$ and $\Theta$ to denote conditional lower bounds under the \ETH.} $f_H(t) = 2^{\Omega(t^c)}$ for some $c > 1$.

Recently, Cygan et al.~\cite{CyganMPP17} started a systematic study of the case $\prec =$ `subgraph', which turns out to exhibit a quite different behavior from the above cases: for every integer $c \geq 1$ there is a graph $H$ such that $f_H(t) = 2^{\Theta(t^c)}$. Cygan et al.~\cite{CyganMPP17} provided a general upper bound and some particular lower bounds on the function $f_H(t)$, but a complete characterization seems to be currently  out of reach.
Previously, Mi. Pilipczuk~\cite{Pilipczuk11} had studied the cases where $H$ is a cycle, finding the function $f_{C_i}(t)$ for every $i \geq 3$ except for $i=4$.

In this article we focus on the case $\prec =$ `induced subgraph' that, to the best of our knowledge, had not been studied before in the literature, except for the case $K_{1,3}$, for which Bonomo-Braberman et al.~\cite{claw-free} showed very recently that $f_{K_{1,3}}(t) = 2^{\Ocal(t^2)}$.

\medskip
\noindent \textbf{Our results and techniques}. We first show (Theorem~\ref{thm:generic-algo-no-colors}) that, for every graph $H$ on $h \geq 3$ vertices, $f_H(t) = 2^{\Ocal(t^{h-2})}$.
 The algorithm uses standard dynamic programming over a nice tree decomposition of the input graph. However, in order to achieve the claimed running time, we need to use a slightly non-trivial encoding in the tables that generalizes an idea of Bonomo-Braberman et al.~\cite{claw-free},  by introducing an object that we call \emph{rooted $H$-folio}, inspired by similar encodings in the context of graph minors~\cite{BasteST20-part1,AdlerDFST11}.

It turns out that when $H$ is a clique or an independent set (in particular, when $|V(H)| \leq 2$), the problem can be solved in {\sl single-exponential} time, that is, $f_H(t) = 2^{\Ocal(t)}$. The case of cliques (Theorem~\ref{Kh_single}), which coincides with the subgraph version, had been already observed by Cygan et al.~\cite{CyganMPP17}, using essentially the folklore fact that every clique is contained in some bag of a tree decomposition. The case of independent sets (Theorem~\ref{thm:algo-Is}) is more interesting, as we exploit tree decompositions in a novel way, by showing (Lemma~\ref{chordal_Ih-free}) that a chordal completion of the complement of a solution can be covered by a constant number of cliques, which implies (Lemma~\ref{bag_cover}) that the complement of a solution is contained in a constant number of bags of the given tree decomposition.

Our main technical contribution consists of lower bounds. Somehow surprisingly, we show (Theorem~\ref{thm:lower-bound-no-colors}) that when $H$ deviates slightly from a clique, the function $f_H(t)$ suffers a sharp jump: if $H$ is obtained from a clique of size $h$ by removing one edge, then $f_H(t) = 2^{\Omega(t^{h-2})}$, and this bound is tight by Theorem~\ref{thm:generic-algo-no-colors}. We also provide lower bounds for other graphs $H$ that are ``close'' to cliques (Theorems~\ref{thm:lower-bound-no-colors-WEAKER},~\ref{thm:lower-bound-no-colors-2}, and~\ref{thm:lower-bound-complete-bipartite}), some of them being (almost) tight. In particular, we show (Theorem~\ref{thm:lower-bound-complete-bipartite}) that when $H=K_{h,h}$, we have that $f_H(t) = 2^{\Omega(t^{h})}$. By observing that the proof of the latter lower bound also applies to occurrences of $K_{h,h}$ as a {\sl subgraph}, the particular case $h=2$ (Corollary~\ref{cor:C4}) answers the open question of Mi. Pilipczuk~\cite{Pilipczuk11} about the function  $f_{C_4}(t)$.  All these reductions are inspired by a reduction of
Cygan et al.~\cite{CyganMPP17} for the subgraph version.
We first present the general frame of the reduction together with some properties that the eventual instances constructed for each of the graphs $H$ have to satisfy, yielding in a unified way (Lemma~\ref{lem:properties}) lower bounds for the corresponding problems.

\medskip

Motivated by the work of Cygan et al.~\cite{CyganMPP17}, we also consider the \emph{colorful} variant of the problem, where the input graph $G$ comes equipped with a coloring $\sigma: V(G) \to V(H)$ and we are only interested in hitting induced subgraphs of $G$ isomorphic to $H$ such that their colors match. In this case, we first observe that essentially the same dynamic programming algorithm of the non-colored version (Theorem~\ref{thm:generic-algo-colors}) yields the upper bound $f_H(t) = 2^{\Ocal(t^{h-2})}$ for every graph $H$ on $h \geq 3$ vertices. Again, our main contribution concerns lower bounds: we show (Theorem~\ref{thm:lower-bound-colors}), by modifying appropriately the frame introduced for the non-colored version, that $f_H(t) = 2^{\Omega(t^{h-2})}$ for every graph $H$ having a connected component on $h$ vertices that is not a clique. Since the case where $H$ is a clique can also be easily solved in single-exponential time (Theorem~\ref{Kh_single}), which can be shown (Theorem~\ref{thm:LBI3}) to be optimal, it follows that  if $H$ is a connected graph on $h \geq 3$ vertices, $f_H(t) = 2^{\Theta(t)}$ if $H$ is a clique, and $f_H(t) = 2^{\Theta(t^{h-2})}$ otherwise. It is easy to see that the cases where $|V(H)| \leq 2$ can be solved in polynomial time by computing a minimum vertex cover in a bipartite graph.

\medskip
\noindent\textbf{Organization.} In Section~\ref{sec:prelim} we provide some basic preliminaries and formally define the problems. In Section~\ref{sec:algos-no-colors} we present the algorithms for both problems, and in Section~\ref{sec:lower-bounds-no-colors} (resp. Section~\ref{sec:with-colors}) we provide the lower bounds for the non-colored (resp. colored) version. Finally, we conclude the article in Section~\ref{sec:concl} with some open questions.
%Due to space limitations, the proofs of the results marked with `$(\star)$' have been moved to the appendices.

\section{Preliminaries}
\label{sec:prelim}

\noindent\textbf{Graphs and functions.} We use standard graph-theoretic notation, and we refer the reader to~\cite{Diestel12} for any undefined notation. We will only consider undirected graphs without loops nor multiple edges, and we denote an edge between two vertices $u$ and $v$ by $\{u,v\}$. A subgraph $H$ of a graph $G$ is \emph{induced} if $H$ can be obtained from $G$ by deleting vertices. A graph $G$ is \emph{$H$-free} if it does not contain any induced subgraph isomorphic to $H$.  For a graph $G$ and a set $S \subseteq V(G)$, we use the notation $G \setminus S :=G[V(G) \setminus S]$. A vertex $v$ is \emph{complete} (resp. \emph{anticomplete}) to a set $S \subseteq V(G)$ if $v$ is adjacent (resp. not adjacent) to every vertex in $S$. A vertex set $S$ of a connected graph $G$ is a \emph{separator} if $G \setminus S$ is disconnected. %and a $(u,v)$-separator, for two vertices $u,v \in V(G)$, if $u$ and $v$

We denote by $\Delta(G)$ (resp. $\omega(G)$) the maximum vertex degree (resp. clique size) of a graph $G$.
For an integer $h \geq 1$, we denote by $P_h$ (resp. $I_h$, $K_h$) the path (resp. independent set, clique) on $h$ vertices, and by $K_h-e$ the graph obtained from $K_h$ by deleting one edge. For two integers $a,b\geq 1$, we denote by $K_{a,b}$ the bipartite graph with parts of sizes $a$ and
$b$. We denote the disjoint union of two graphs $G_1$ and $G_2$ by $G_1 + G_2$.

A graph property is \emph{hereditary} if whenever it holds for a graph $G$, it holds for all its induced subgraphs as well. The \emph{open} (resp. \emph{closed}) \emph{neighborhood} of a vertex $v$ is denoted by $N(v)$ (resp. $N[v]$). A vertex is \emph{simplicial} if its (open or closed) neighborhood induces a clique. A graph $G$ is \emph{chordal} if it does not contain induced cycles of length at least four or, equivalently, if $V(G)$ can be ordered $v_1, \ldots, v_n$ such that, for every $2 \leq  i \leq  n$, vertex $v_i$ is simplicial in the subgraph of $G$ induced by $\{v_1, \ldots, v_{i-1}\}$. Note that being chordal is a hereditary property.

\smallskip
Given a function $f:A \to B$ between two sets $A$ and $B$ and a subset $A' \subseteq A$, we denote by $f|_{A'}$ the restriction of $f$ to $A'$ and by ${\sf im}(f)$ the image of $f$, that is, ${\sf im}(f) = \{b \in B \mid \exists a \in A: f(a)=b\}$. For an integer $k \geq 1$, we let $[k]$ be the set containing all integers $i$ with $1 \leq i \leq k$.

%A \emph{star} is a tree that has at most one vertex of degree larger than one, called the \emph{center} of the star.
% $k$-trees

%Parameterized complexity

\medskip

\noindent\textbf{Definition of the problems.} Before formally defining the problems considered in this article, we introduce some terminology, mostly taken from~\cite{CyganMPP17}. Given a graph $H$, an \emph{$H$-coloring} of a graph $G$ is a function $\sigma: V(G) \to V(H)$. A \emph{homomorphism} (resp. \emph{induced homomorphism} from a graph $H$ to a graph $G$ is a function $\pi: V(H) \to V(G)$ such that $\{u,v\} \in E(H)$ implies (resp. if and only if) $\{\pi(u),\pi(v)\} \in E(G)$. When $G$ is $H$-colored by a function $\sigma$, an \emph{(induced) $\sigma$-homomorphism} from $H$ to $G$ is an (induced) homomorphism $\pi$ from $H$ to $G$ with the additional property that every vertex is mapped to the appropriate color, that is, $\sigma(\pi(a))=a$ for every vertex $a \in V(H)$. An \emph{(induced) $H$-subgraph} of $G$ is an (induced) injective homomorphism from $H$ to $G$ and, if $G$ is $H$-colored by a function $\sigma$, an \emph{(induced) $\sigma$-$H$-subgraph} of $G$ is an (induced) injective $\sigma$-homomorphism from $H$ to $G$. We say that a vertex set $X \subseteq V(G)$ \emph{hits} an (induced) $\sigma$-$H$-subgraph $\pi$ if $X \cap \pi(V(H)) \neq \emptyset$.

For a fixed graph $H$, the problems we consider in this article are defined as follows.
\vspace{-.3cm}

\probl{$H$-IS-Deletion}{A graph $G$.}{The minimum size of a set $X \subseteq V(G)$ that hits all induced $H$-subgraphs of $G$. }

\probl{Colorful $H$-IS-Deletion}{A graph $G$ and an $H$-coloring $\sigma$ of $G$.}{The minimum size of a set $X \subseteq V(G)$ that hits all induced $\sigma$-$H$-subgraphs of $G$. }

%Given an input graph $G$ (and an $H$-coloring $\sigma$ of $G$), the \textsc{(Colorful) $H$-IS-Deletion} problem asks for the minimum size of a set $X \subseteq V(G)$ that hits all induced $(\sigma$-)$H$-subgraphs of $G$.

The \textsc{$H$-S-Deletion} and \textsc{Colorful $H$-S-Deletion} problems are defined similarly, just by removing the word `induced' from the above definitions. In the decision version of these problems, we are given a target budget $k$, and the objective is to decide whether there exists a hitting set of size at most $k$. Unless stated otherwise, we let $n$ denote the number of vertices of input graph of the problem under consideration. When expressing the running time of an algorithm, we will sometimes use the
$\Ocal^*(\cdot)$ notation, which  suppresses polynomial factors in the input size.

%\ig{Define \emph{$H$-hitting set}}

\medskip

\noindent\textbf{Tree decompositions.} A \emph{tree decomposition} of a graph $G$ is a pair ${\cal D}=(T,{\cal X})$, where $T$ is a tree
and ${\cal X}=\{X_{w}\mid w\in V(T)\}$ is a collection of subsets of $V(G)$, called \emph{bags},
% indexed by the vertices of $T$,
such that:
\begin{itemize}
\item $\bigcup_{w \in V(T)} X_w = V(G)$,
\item for every edge $\{u,v\} \in E$, there is a $w \in V(T)$ such that $\{u, v\} \subseteq X_w$, and
\item for each $\{x,y,z\} \subseteq V(T)$ such that $z$ lies on the unique path between $x$ and $y$ in $T$,  $X_x \cap X_y \subseteq X_z$.
\end{itemize}
We call the vertices of $T$ {\em nodes} of ${\cal D}$ and the sets in ${\cal X}$ {\em bags} of ${\cal D}$. The \emph{width} of a  tree decomposition ${\cal D}=(T,{\cal X})$ is $\max_{w \in V(T)} |X_w| - 1$.
The \emph{treewidth} of a graph $G$, denoted by $\tw(G)$, is the smallest integer $t$ such that there exists a tree decomposition of $G$ of width at most $t$.
%For each $t \in V(T)$, we denote by $E_w$ the set $E(G[X_w])$ \ig{do we use it?}.
We need to introduce nice tree decompositions, which will make the presentation of the algorithms much simpler. %\ig{really? Only if we write a formal proof!}

%\ig{Make link with $k$-trees}

\medskip
\noindent
\textbf{Nice tree decompositions.} Let ${\cal D}=(T,{\cal X})$
be a tree decomposition of $G$, $r$ be a vertex of $T$, and   ${\cal G}=\{G_{w}\mid w\in V(T)\}$ be
a collection of subgraphs of   $G$, indexed by the vertices of $T$.
A triple $({\cal D},r,{\cal G})$ is a
\emph{nice tree decomposition} of $G$ if the following conditions hold:
\begin{itemize}

\item $X_{r} = \emptyset$ and $G_{r}=G$,
\item each node of ${\cal D}$ has at most two children in $T$,
\item for each leaf $\ell \in V(T)$, $X_{\ell} = \emptyset$ and $G_{\ell}=(\emptyset,\emptyset).$ Such an $\ell$ is called a {\em leaf node},
\item if $w \in V(T)$ has exactly one child $w'$, then either
  \begin{itemize}
  \item $X_w = X_{w'}\cup \{v_{\rm in}\}$ for some $v_{\rm in} \not \in X_{w'}$ and $G_{w}=G[V(G_{w'})\cup\{v_{\rm in}\}]$.
    The node $w$ is called an \emph{introduce node}  and the vertex $v_{\rm in}$ is the {\em introduced vertex} of $X_{w}$,
    % \item $X_w = X_{w'}$ and $G_{w}=(G_{w'},E(G_{w'})\cup\{e_{\rm insert}\})$ where $e_{\rm insert}$ is an edge of $G$ with endpoints in  $X_{w}$.
    %   The node $t$ is called \emph{introduce edge} node  and the edge $e_{\rm insert}$ is the {\em insertion edge} of $X_{w}$, or
  \item $X_w = X_{w'} \setminus \{v_{\rm out}\}$ for some $v_{\rm out} \in X_{w'}$ and $G_{w}=G_{w'}$.
    The node $w$ is called a  \emph{forget node} node and $v_{\rm out}$ is the {\em forget vertex} of $X_{w}$.
  \end{itemize}
\item if $w \in V(T)$ has exactly two children $w_1$ and $w_2$, then $X_{w} = X_{w_1} = X_{w_2}$, $E(G_{w_1})\cap E(G_{w_2})=E(G[X_w])$, and $G_w = (V(G_{w_1})\cup V(G_{w_2}),E(G_{w_1})\cup E(G_{w_2}))$. The node $w$ is called a \emph{join node}.
\end{itemize}

% \todo{This next paragraph needs to be changed. I have removed the introduce edge.}
% The notion of a nice triple defined above is essentially the same
% as the one of nice tree decomposition in~\cite{CyganNPPRW11} (which in turn is an enhancement of the original one, introduced in~\cite{Klo94}).

For each $w \in V(T)$, we denote by $V_w$ the set $V(G_w)$.
Given a tree decomposition, it is possible to transform it in polynomial time to a {\sl nice} one of the same width~\cite{Klo94}. Moreover, by Bodlaender et al.~\cite{BodlaenderDDFLP16} we can find in time $2^{\Ocal(\tw)}\cdot n$ a tree decomposition of width $\Ocal(\tw)$ of any graph $G$ with treewidth $\tw$. Since the running time of our algorithms dominates this function, we may assume that a nice tree decomposition of width $t = \Ocal(\tw)$
% \jul{$\w = \Ocal(\tw)$}
is given along with the input. %\ig{do we use it?}

\medskip
\noindent
\textbf{Exponential Time Hypothesis.} The \emph{Exponential Time Hypothesis} (\ETH) of Impagliazzo and Paturi~\cite{ImpagliazzoP01-ETH} implies that the 3-\textsc{Sat} problem on $n$ variables cannot be solved in time $2^{o(n)}$.
The Sparsification Lemma of Impagliazzo et al.~\cite{ImpagliazzoP01} implies that if the \ETH holds, then there is no algorithm solving a 3-\textsc{Sat} formula with $n$ variables and $m$ clauses in time $2^{o(n+m)}$. Using the terminology from Cygan et al.~\cite{CyganMPP17}, a 3-\textsc{Sat} formula $\varphi$ in conjunctive normal form is said to be \emph{clean} if each variable of $\varphi$ appears exactly three times, at least once positively and at least once negatively, and each clause of $\varphi$ contains two or three literals and does not contain twice the same variable. Cygan et al.~\cite{CyganMPP17} observed the following useful lemma.

\begin{lemma}[Cygan et al.~\cite{CyganMPP17}]\label{lem:clean3SAT}
The existence of an algorithm in time $2^{o(n)}$ deciding whether a clean 3-\textsc{Sat} formula with $n$ variables is satisfiable would violate the \ETH.
\end{lemma}

\section{Algorithms}
\label{sec:algos-no-colors}

In this section we present algorithms for \textsc{$H$-IS-Deletion} and \textsc{Colorful $H$-IS-Deletion}. We start in Subsection~\ref{sec:generic-algo} with a general dynamic programming algorithm that solves \textsc{$H$-IS-Deletion} and \textsc{Colorful $H$-IS-Deletion} in time $\Ocal^*(2^{\Ocal(t^{h-2})})$ for any graph $H$ on at least $h \geq 3$ vertices.
In Subsection~\ref{sec:cliques-indep-sets} we focus on hitting cliques and independent sets.

%As observed by Cygan et al.~\cite{CyganMPP17}, both \textsc{$H$-IS-Deletion} and \textsc{Colorful $H$-IS-Deletion} can be solved in time $\Ocal^*(2^{\Ocal(t)})$ when $H$ is a clique. We present an algorithm to solve
%\textsc{$H$-IS-Deletion} in time $\Ocal^*(2^{\Ocal(t)})$ when $H$ is an independent set.

\subsection{A general dynamic programming algorithm}
\label{sec:generic-algo}

We present the algorithm for \textsc{$H$-IS-Deletion}, and then we discuss that essentially the same algorithm applies to \textsc{Colorful $H$-IS-Deletion} as well. Our algorithm to solve \textsc{$H$-IS-Deletion}  in time $\Ocal^*(2^{\Ocal(t^{h-2})})$ uses standard dynamic programming over a nice tree decomposition of the input graph; we refer the reader to~\cite{CyganFKLMPPS15} for a nice exposition. However, in order to achieve the claimed running time, we need to use a slightly non-trivial encoding in the tables, which we proceed to explain.

Let $|V(H)|=h$ and assume that we are given a nice tree decomposition of the input graph $G$ such that its bags contain at most $t$ vertices (in a tree decomposition of width $t$, the bags have size at most $t+1$, but to simplify the exposition we assume that they have size at most $t$, which does not change the asymptotic complexity of the algorithm).

Intuitively, our algorithm proceeds as follows. At each bag $X_w$ of the nice tree decomposition of $G$, a state is indexed by the intersection of the desired hitting set constructed so far with the bag, and the collection of {\sl proper} subgraphs of $H$ that occur as induced subgraphs in the graph obtained from $G_w$ after removing the current solution. In order to be able to proceed with the dynamic programming routine while keeping the complement of the hitting set $H$-free, we need to remember how these proper subgraphs of $H$ intersect with $X_w$, and this is the most expensive part of the algorithm in terms of running time. We encode this collection of rooted subgraphs (where the ``roots'' correspond to the vertices in $X_w$) of $H$ with an object $\Hcal_w$ that we call a \emph{rooted $H$-folio}, inspired by similar encodings in the context of graph minors~\cite{BasteST20-part1,AdlerDFST11}. Since we need to remember proper subgraphs of $H$ on at most $h-1$ vertices, and we have up to $t$ choices to root each of their vertices in the bag $X_w$, the number of rooted proper subgraphs of $H$ is at most $t^{h-1}$. Therefore, the number of rooted $H$-folios, each corresponding to a collection of rooted proper subgraphs of $H$, is bounded from above by $2^{t^{h-1}}$. This encoding naturally leads to a dynamic programming algorithm to solve \textsc{$H$-IS-Deletion}  in time $\Ocal^*(2^{\Ocal(t^{h-1})})$, where the hidden constants (but not the degree of the polynomial in $n$) may depend on $H$.

In order to further reduce the exponent to $h-2$, we use the following trick inspired by the dynamic programming algorithm of Bonomo-Braberman et al.~\cite{claw-free} to solve \textsc{$K_{1,3}$-IS-Deletion} in time $\Ocal^*(2^{\Ocal(t^{2})})$. The crucial observation is the following: the existence of proper induced subgraphs of $H$ that are {\sl fully} contained in the current bag $X_w$ can be checked {\sl locally} within that bag, without needing to root their vertices. That is, we distinguish these {\sl local occurrences} of proper induced subgraphs of $H$, and we encode them separately in $\Hcal_w$, without rooting their vertices in $X_w$. Note that the number of choices for those local occurrences depends only on $H$. In particular, since the proper subgraphs of $H$ have at most $h-1$ vertices, the previous observation implies that we never need to root exactly $h-1$ vertices of an induced subgraph of $H$, since such occurrences would be fully contained in $X_w$. This permits to improve the running time to $\Ocal^*(2^{\Ocal(t^{h-2})})$. The details follow.

Note that we may assume that $H$ has at least three vertices, as otherwise it is a clique or an independent set, and then \textsc{$H$-IS-Deletion} can be solved in single-exponential time by the algorithms in Subsection~\ref{sec:cliques-indep-sets}.

\begin{theorem}\label{thm:generic-algo-no-colors}
For every graph $H$ on $h\geq 3$ vertices, the {\sc $H$-IS-Deletion} problem can be solved in time $2^{\Ocal(t^{h-2})} \cdot n$, where $n$ and $t$ are the number of vertices and the treewidth of the input graph, respectively.
\end{theorem}

\begin{proof}
 As discussed in Section~\ref{sec:prelim}, we may assume that we are given are {\sl nice} tree decomposition $({\cal D},r,{\cal G})$ of $G$ of width $\Ocal(\tw(G))$. To simplify the exposition, suppose that the bags of ${\cal D}$ contain at most $t$ vertices, and recall that this assumption does not change the asymptotic complexity of the algorithm.

At each bag $X_w$ of the given nice tree decomposition of $G$, a valid state of our dynamic programming table is indexed by the following two objects:

\begin{itemize}
\item A subset $\hat{S}_w \subseteq X_w$ that corresponds to the intersection of the desired $H$-hitting set with the current bag. With this in mind, we say that a set $S_w \subseteq V_w$ is \emph{feasible for $\hat{S}_w$} if $G_w \setminus S_w$ is $H$-free and $S_w \cap X_w = \hat{S}_w$.
\item A rooted $H$-folio $\Hcal_w$, containing the following two collections of (rooted) proper induced subgraphs of $H$:

    \begin{itemize}
    \item The set $\Lcal_w$ of \emph{local} occurrences of proper induced subgraphs of $H$, consisting of the set of proper induced subgraphs of $H$ that occur in $G[X_w \setminus \hat{S}_w]$.
    \item The set $\Rcal_w$ of \emph{rooted} occurrences of proper induced subgraphs of $H$, consisting of a set of triples $(\tilde{H},R,\rho)$ where $\tilde{H}$ is a proper induced subgraph of $H$, $R \subseteq V(\tilde{H})$ is a set with $0 \leq |R| \leq h-2$ corresponding to the {\sl roots} of $\tilde{H}$ in $X_w \setminus \hat{S}_w$, and $\rho: R \to X_w \setminus \hat{S}_w$ is an injective function that maps each vertex in $R$ to its corresponding vertex in $X_w \setminus \hat{S}_w$.
    \end{itemize}
\end{itemize}

\noindent Note that the number of local occurrences of proper induced subgraphs of $H$ depends only on $H$, and that the number of tuples $(\tilde{H},R,\rho)$ of rooted occurrences is at most $2^{h} \cdot 2^{h-1}\cdot t^{h-2}$, and therefore the number of rooted $H$-folios is at most $2^{\Ocal(t^{h-2})}$, as desired.

We say that the rooted $H$-folio of a subgraph $G'_w \subseteq G_w$ is $\Hcal_w$ if the local and rooted occurrences of induced subgraphs of $H$ in $G'_w$ correspond exactly to the collections $\Lcal_w$ and $\Rcal_w$ of $\Hcal_w$, respectively. Our algorithm stores, for each state $(\hat{S}_w, \Hcal_w)$ of a node $w$ of a nice tree decomposition of $G$, the minimum size of a set $S_w \subseteq V_w$ feasible for $\hat{S}_w$ such that the rooted $H$-folio of $G_w \setminus S_w$ is  $\Hcal_w$, or $+\infty$ if such a set does not exist. We denote this value by ${\sf opt}(\hat{S}_w, \Hcal_w)$.

When $r$ is the root of the nice tree decomposition, note that the solution of the \textsc{$H$-IS-Deletion} problem in $G$ equals $\min_{\Hcal_r}\{ {\sf opt}(\emptyset, \Hcal_r)  \}$, where $\Hcal_r = (\Lcal_r , \Rcal_r )$ runs over all rooted $H$-folios such that $\Lcal_r = \emptyset$
and $ \Rcal_r$ contains the triples $(\tilde{H},\emptyset,\emptyset)$ for all proper induced subgraphs $\tilde{H}$ of $H$.

\medskip

We now show how these valid states and associated values can be computed recursively in a typical bottom-up fashion starting from the leaves, by distinguishing the distinct types of nodes in a nice tree decomposition. We let $V(H)= \{z_1,\ldots,z_h\}$. %\ig{do we need it?}

\medskip
\noindent\textbf{Leaf node}. The unique valid state is $(\emptyset, \emptyset)$ and ${\sf opt}(\emptyset, \emptyset)= 0$.

\medskip
\noindent\textbf{Introduce node}. Let $w$ be an introduce node with child $w'$ such that $X_w \setminus X_{w'} = \{v\}$. For each valid state $(\hat{S}_{w'}, \Hcal_{w'})$ for $w'$, with $\Hcal_{w'} = (\Lcal_{w'} , \Rcal_{w'})$, we generate the following valid states for $w$, depending on whether $v$ is included in the current partial hitting set in $X_w$ or not:

\begin{itemize}
\item $(\hat{S}_{w'} \cup \{v\}, \Hcal_{w'})$. In this case, we just include $v$ into the partial hitting set, hence the rooted $H$-folio remains the same. Therefore, ${\sf opt}(\hat{S}_{w'} \cup \{v\}, \Hcal_{w'}) = {\sf opt}(\hat{S}_{w'}, \Hcal_{w'}) + 1$.
\item $(\hat{S}_{w'}, \Hcal_{w})$ only if  $G[X_w \setminus \hat{S}_{w'}]$ is $H$-free, where $\Hcal_w = (\Lcal_w , \Rcal_w)$ is defined as follows:
\begin{itemize}
\item $\Lcal_w $ contains all the proper induced subgraphs of $H$ that occur in $G[X_w \setminus \hat{S}_{w'}]$. Note that this set can be computed in time $\Ocal(t^h)$, where the hidden constant depends on $h$.
    %, we discard this state $(\hat{S}_{w'}, \Hcal_{w})$ from the table of $X_w$.
\item In order to define the set of triples contained in $\Rcal_w$, we first check that $H$ does not occur when introducing $v$: if for some triple $(\tilde{H}',R',\rho') \in \Rcal_{w'}$, the graph obtained from $\tilde{H}'$ by adding a new vertex $u$ and an edge $\{u,z\}$ for every root vertex $z \in R'$ such that $\rho'(z)$ is adjacent to $v$ in $G$, is isomorphic to $H$, we discard the  state $(\hat{S}_{w'}, \Hcal_{w})$ from the table of $X_w$, and we move on to the next state for $w'$.
    If there is no such triple, we add the whole collection $\Rcal_{w'}$ to $\Rcal_w$. Moreover, we add to $\Rcal_w$ every triple $(\tilde{H},R,\rho)$  that can be obtained from a triple $(\tilde{H}',R',\rho') \in \Rcal_{w'}$ with $|V(\tilde{H}')| \leq h-2$, $|R'| \leq h-3$, $v \in {\sf im}(\rho)$, $|R|= |R'| + 1$, $\rho|_{R'}= \rho'$, $\tilde{H}$ is a proper induced subgraph of $H$, and $\tilde{H}'$ is isomorphic to $\tilde{H} \setminus \{ \rho^{-1}(v)\}$. That is, since in this case $v$ does not belong to the partial hitting set, we also add to $\Rcal_w$ any rooted occurrence that can be obtained from a rooted occurrence in $\Rcal_{w'}$ by adding vertex $v$ to the set of roots $R$ to form a larger $\tilde{H}$.
\end{itemize}
In this case we set ${\sf opt}(\hat{S}_{w'}, \Hcal_{w}) = {\sf opt}(\hat{S}_{w'}, \Hcal_{w'})$.
\end{itemize}

\medskip
\noindent\textbf{Forget node}. Let $w$ be a forget node with child $w'$ such that $X_{w'} \setminus X_{w} = \{v\}$. For each valid state $(\hat{S}_{w'}, \Hcal_{w'})$ for $w'$, with $\Hcal_{w'} = (\Lcal_{w'} , \Rcal_{w'})$, we generate the following valid states for $w$, depending on whether $v \in \hat{S}_{w'}$ or not:

\begin{itemize}
\item If $v \in \hat{S}_{w'}$ we add the state $(\hat{S}_{w'} \setminus \{v\}, \Hcal_{w'})$ and we set ${\sf opt}(\hat{S}_{w'} \setminus \{v\}, \Hcal_{w'}) = {\sf opt}(\hat{S}_{w'}, \Hcal_{w'})$. In this case, we just forget vertex $v$, which was in the solution, and $\Hcal_{w'}$ remains the same.

\item Otherwise, if $v \notin \hat{S}_{w'}$, we add the state $(\hat{S}_{w'}, \Hcal_w)$   where $\Hcal_w = (\Lcal_w , \Rcal_w)$ is defined as follows:
    \begin{itemize}
\item $\Lcal_w $ contains all the proper induced subgraphs of $H$ that occur in $G[X_w \setminus \hat{S}_{w'}]$. Again, this set can be computed locally in time  $\Ocal(t^h)$.

\item  $\Rcal_w$ contains every triple $(\tilde{H},R,\rho)$ that can be constructed by any of the following two operations:
    \begin{itemize}
    \item If there is a local occurrence $\tilde{H}' \in \Lcal_{w'}$ such that $v$ belongs to an induced $\tilde{H}'$-subgraph $F$ of $G[X_{w'} \setminus \hat{S}_{w'}]$, we add to $\Rcal_w$ the triple $(\tilde{H},R,\rho)$ defined as $\tilde{H}=\tilde{H}'$, $R = V(\tilde{H}') \setminus \{z\}$ where $z$ is the vertex of $\tilde{H}'$ mapped to $v$, and $\rho$ mapping every vertex of $R$ to their image in $F$. That is, if $v$ was part of a local occurrence for node $w'$, now this occurrence becomes a rooted one for node $w$, defined in the natural way.
    \item Let $(\tilde{H}',R',\rho') \in \Rcal_{w'}$ be a rooted occurrence in $\Rcal_{w'}$. We distinguish two cases:
    \begin{itemize}
    \item If $v \notin {\sf im}(\rho')$, we add $(\tilde{H}',R',\rho')$  to $\Rcal_w$.
    \item Otherwise, if $v \in {\sf im}(\rho')$, we add to to $\Rcal_w$ the triple $(\tilde{H},R,\rho)$ defined as $\tilde{H} = \tilde{H}'$, $R = R' \setminus \{\rho'^{-1}(v)\}$, and $\rho = \rho'|_{R}$. That is, we just remove the forgotten vertex $v$ from the root set of the corresponding rooted occurrence.
    \end{itemize}
    \end{itemize}
\end{itemize}
In this case we set ${\sf opt}(\hat{S}_{w'}, \Hcal_{w}) = {\sf opt}(\hat{S}_{w'}, \Hcal_{w'})$.
\end{itemize}

\medskip
\noindent\textbf{Join node}. Let $w$ be a join node with children $w_1$ and $w_2$. For each pair of valid states $(\hat{S}_{w_1}, \Hcal_{w_1})$ and $(\hat{S}_{w_2}, \Hcal_{w_2})$ for $w_1$ and $w_2$,  with $\Hcal_{w_1} = (\Lcal_{w_1} , \Rcal_{w_1})$ and $\Hcal_{w_2} = (\Lcal_{w_2} , \Rcal_{w_2})$, respectively, such that $\hat{S}_{w_1} = \hat{S}_{w_2}$ and $\Lcal_{w_1} = \Lcal_{w_2}$, we generate the valid state $(\hat{S}_{w_1}, \Hcal_{w})$ for $w$, where $\Hcal_{w} = (\Lcal_{w_1} , \Rcal_{w})$ and $\Rcal_{w}$ is defined as follows. For every pair of rooted triples $(\tilde{H}_1,R_1,\rho_1) \in \Rcal_{w_1}$ and $(\tilde{H}_2,R_2,\rho_2) \in \Rcal_{w_2}$ with $R_1 = R_2$, $\rho_1 = \rho_2$, and $(V(\tilde{H}_1) \setminus R_1) \cap (V(\tilde{H}_2) \setminus R_2) = \emptyset$ (recall that the vertices of $H$ are labeled, so this condition is well-defined), we add to $\Rcal_{w}$ the triple $(\tilde{H}_1 \cup \tilde{H}_2,R_1,\rho_1)$. That is, we just merge the rooted triples that coincide in $X_w = X_{w_1} = X_{w_2}$, by taking the union of the corresponding subgraphs of $H$.

Finally, for $(\hat{S}_{w_1}, \Hcal_{w})$ to be indeed a valid state for $w$, we have to check that an occurrence of $H$ has not been created  in those triples: if for some such a triple $(\tilde{H}_1 \cup \tilde{H}_2,R_1,\rho_1) \in \Rcal_{w}$, we have that $\tilde{H}_1 \cup \tilde{H}_2 = H$, we discard the state $(\hat{S}_{w_1}, \Hcal_{w})$ for $w$, and we move on to the next pair of valid states $(\hat{S}_{w_1}, \Hcal_{w_1})$ and $(\hat{S}_{w_2}, \Hcal_{w_2})$ for $w_1$ and $w_2$, respectively.

If the state $(\hat{S}_{w_1}, \Hcal_{w})$ has not been discarded, we set
$$
{\sf opt}(\hat{S}_{w_1}, \Hcal_{w}) = {\sf opt}(\hat{S}_{w_1}, \Hcal_{w_1}) + {\sf opt}(\hat{S}_{w_2}, \Hcal_{w_2}) - |\hat{S}_{w_1}|.
$$

\bigskip

In all cases, if distinct valid states of the child(ren) node(s) generate the same valid state $(\hat{S}_{w}, \Hcal_{w})$ at the current node $w$, we update ${\sf opt}(\hat{S}_{w}, \Hcal_{w})$ to be the minimum among all the obtained values, as usual. This concludes the description of the algorithm, whose correctness follows from the definition of the tables and the fact that the solution of the \textsc{$H$-IS-Deletion} problem on $G$ is computed at the root of the nice tree decomposition. Clearly, all the above operations can be performed at each node in time $2^{\Ocal(t^{h-2})}$, and the proof is complete by taking into account that we may assume that the given nice tree decomposition has $\Ocal(t \cdot n)$ nodes~\cite{Klo94}.
\end{proof}

A dynamic programming algorithm similar to the one provided in Theorem~\ref{thm:generic-algo-no-colors} can also solve the \textsc{Colorful $H$-IS-Deletion} problem  in time $2^{\Ocal(t^{h-2})} \cdot n$ for every graph $H$ on $h \geq 3$ vertices. Indeed, the algorithm remains basically the same, except that we have to keep track only of {\sl colorful} copies of proper subgraphs of $H$, and to discard only the states in which a {\sl colorful} occurrence of $H$ appears. In order to do that, in the tables of the dynamic programming algorithm we just need to replace rooted $H$-folios by \emph{rooted $\sigma$-$H$-folios}, defined in the natural way. Since the number of further computations at each node in order to verify that the colors match in the obtained rooted subgraphs of $H$  is a function dominated by $2^{\Ocal(t^{h-2})}$, we obtain the same asymptotic running time. We omit the details.

\begin{theorem}\label{thm:generic-algo-colors}
For every graph $H$ on $h\geq 3$ vertices, the {\sc Colorful $H$-IS-Deletion} problem can be solved in time $2^{\Ocal(t^{h-2})} \cdot n$, where $n$ and $t$ are the number of vertices and the treewidth of the input graph, respectively.
\end{theorem}

It is easy to check  that small adaptations of the algorithms of Theorems~\ref{thm:generic-algo-no-colors} and~\ref{thm:generic-algo-colors} also work for the (not necessarily induced) subgraph version of both problems. Nevertheless, the obtained running times never outperform those obtained by Cygan et al.~\cite{CyganMPP17} for those problems.

\subsection{Hitting cliques and independent sets}
\label{sec:cliques-indep-sets}

%In this section we will assume that the reader is familiar with dynamic programmings in graphs with bounded treewidth using tree decompositions, for more details see~\cite{CyganFKLMPPS15}. We consider that the reader is acquainted with standard dynamic programmings to solve vertex deletion problems parameterized by treewidth, so in this section we focus on what is actually relevant to be stored, assuming that for the reader standard dynamic programmings for \textsc{$H$-IS-Deletion} (also known as $H$-free Vertex Deletion) is an exercise.

The following folklore lemma follows easily from the definition of tree decomposition.

\begin{lemma}\label{clique_bag}
Let $G$ be a graph and let ${\cal D}$ be a tree decomposition of $G$. Then every clique of $G$ is contained in some bag of ${\cal D}$.
\end{lemma}

Note that if $H$ is a clique, then the \textsc{(Colorful) $H$-IS-Deletion} problem is the same as the \textsc{(Colorful) $H$-S-Deletion} problem. Cygan et al.~\cite{CyganMPP17} observed that, by Lemma~\ref{clique_bag}, in order to solve \textsc{(Colorful) $K_h$-IS-Deletion} it is enough to do the following: store, for every bag of a (nice) tree decomposition of the input graph, the subset of vertices of the bag that belongs to the partial hitting set, and check locally within the bag that the remaining vertices do not induce a $K_h$. A typical dynamic programming routine yields the following result\footnote{In fact, Cygan et al.~\cite{CyganMPP17} presented an algorithm only for \textsc{Colorful $K_h$-S-Deletion}, but the algorithm for \textsc{$K_h$-S-Deletion} is just a simplified version of  the colorful version, just by forgetting the colors.}.

%From Lemma~\ref{clique_bag} follows that we do not need to store additional information in order to solve \textsc{$K_h$-IS-Deletion}: to compute a minimum solution, for each bag, it is enough to store all feasible set of vertices that may be part of the solution (removed), where a set is feasible if the remaining vertices of the bag do not induce a $K_h$; thus, using the nodes of the tree decomposition and their feasible sets as table indexes, in a bottom-up manner, one can proceed with the computation. This implies, as also observed by Cygan et al.~\cite{CyganMPP17},  that the following holds.

\begin{theorem}[Cygan et al.~\cite{CyganMPP17}]
\label{Kh_single}
 For every integer $h\geq 1$, {\sc $K_h$-IS-Deletion} and {\sc Colorful $K_h$-IS-Deletion} can be solved in time $2^{\Ocal(t)} \cdot n$, where $n$ and $t$ are the number of vertices and the treewidth of the input graph, respectively.
\end{theorem}

The case where $H$ is an independent set, which is \NP-hard by~\cite{LeYa80},  turns out to be more interesting. We proceed to present a single-exponential algorithm for \textsc{$I_h$-IS-Deletion}, and we remark that this algorithm does {\sl not} apply to the colorful version.

%A celebrated result of Lewis and Yannakakis~\cite{L-Y-node-deletion} shows that for
%any hereditary and nontrivial graph class $\mathcal{C}$, the problem of determining the minimum number of vertices needed to be removed to make the resulting graph belong to $\mathcal{C}$ is $\NP$-hard. Since $H$-free is a hereditary and nontrivial graph class, for any $H$ on $h\geq 2$ vertices, it holds that \textsc{$H$-IS-Deletion} is NP-hard for any $H$ of size at least two, and in particular, there is no $n^{f(h)}$-time algorithm for \textsc{$H$-IS-Deletion}.
%
%To the best of our knowledge, the only infinite family $\mathcal F$, known in the literature, for which  for every  $H \in \mathcal F$, \textsc{$H$-IS-Deletion} can be solved in single-exponential time with respect to the treewidth, is the family of complete graphs.

%Next, we consider the family of empty graphs.

Note that \textsc{$I_2$-IS-Deletion} is the dual problem of {\sc Maximum Clique}, since a minimum $I_2$-hitting set is the complement of a maximum clique. This duality together with Lemma~\ref{clique_bag} yield the following key insight: in any graph $G$, after the removal of an optimal solution of \textsc{$I_2$-IS-Deletion}, all the remaining vertices are contained in a single bag of any tree decomposition of $G$. Our algorithm is based on a generalization of this property to any $h \geq 1$, stated in Lemma~\ref{bag_cover}, which gives an alternative way to exploit tree decompositions in order to solve the \textsc{$H$-IS-Deletion} problem.

We first need a technical lemma. A \emph{clique cover} of a graph $G$ is a collection of cliques of $G$ that cover $V(G)$, and its \emph{size} is the number of cliques in the cover.

%Motivated by this observation, we investigated the generalization of such a property for \textsc{$I_h$-IS-Deletion}, and as a by-product we obtain an alternative way to explore the tree decomposition for solving these vertex deletion problems.

\begin{lemma}\label{chordal_Ih-free}
Every $I_h$-free chordal graph $G$ admits a clique cover of size at most $h-1$.
\end{lemma}
\begin{proof}
We prove the lemma by induction on $h$. For $h=2$, $G$ itself is a clique and the claim is trivial.
Suppose inductively that any $I_{h-1}$-free chordal graph admits a clique cover of size at most $h-2$, let $G$ be an $I_h$-free chordal graph, and let $v$ be a simplicial vertex of $G$. Since $N[v]$ is a clique and $G$ is $I_h$-free, it follows that $G \setminus N[v]$ is $I_{h-1}$-free.
Since being chordal is a hereditary property,
$G \setminus N[v]$  is an $I_{h-1}$-free chordal graph, so by induction $G \setminus N[v]$ admits a  clique cover of size at most $h-2$.  These $h-2$ cliques together with $N[v]$ define a clique cover of $G$ of size at most $h-1$.
\end{proof}

\begin{lemma}\label{bag_cover}
Let $h\geq 2$ be an integer, let $G$ be a graph, let ${\cal D}$ be a tree decomposition of $G$, and let $S$ be any solution for \textsc{$I_h$-IS-Deletion} on $G$.
Then there are at most $h-1$ bags $X_1,X_2,\ldots,X_{h-1}$ of ${\cal D}$ such that $V(G)\setminus S\subseteq \bigcup_{i \in [h-1]}X_{i}$.
\end{lemma}

\begin{proof}
Let ${\cal D}$ be a tree decomposition of $G$, let $S$ be a solution for \textsc{$I_h$-IS-Deletion} on $G$, and let $G^{\star}$ be the graph obtained from $G$ by adding an edge between any pair of vertices contained in the same bag of ${\cal D}$.
Note that $G^{\star}$ is a chordal graph, and that ${\cal D}$ is also a tree decomposition of $G^{\star}$.
%And, by removing vertices of $S$ of each bag, $(T\setminus S)$ forms a tree decomposition of $G^{\star}\setminus S$.
Since being a chordal graph is a hereditary property, it follows that $G^{\star}\setminus S$ is chordal. Since $G \setminus S$ is $I_h$-free, and the property of being $I_h$-free is closed under edge addition, we have that $G^{\star} \setminus S$ is also $I_h$-free. Thus, $G^{\star}\setminus S$ is an $I_h$-free chordal graph, and Lemma~\ref{chordal_Ih-free} implies that $G^{\star}\setminus S$ admits a clique cover of size at most $h-1$. Since any clique in $G^{\star}\setminus S$ is also a clique in $G^{\star}$, and ${\cal D}$ is a tree decomposition of $G^{\star}$,  Lemma~\ref{clique_bag}  implies that every clique of  $G^{\star}\setminus S$ is contained in some bag of ${\cal D}$, and therefore there are at most $h-1$ bags of ${\cal D}$ that cover all vertices in $V(G^{\star}) \setminus S = V(G)\setminus S$.
%%%%%%%%%%%%%%%%%%%%
\end{proof}

Recall that \textsc{$I_h$-IS-Deletion}  is \NP-hard even for $h=2$, thus the problem cannot be solved in time $n^{f(h)}$ for any function $f$, unless ${\sf P}=\NP$.

\begin{theorem}\label{thm:algo-Is}
 For every integer $h\geq 1$, {\sc $I_h$-IS-Deletion} can be solved in time $2^{\Ocal(t)} \cdot n^{h}$, where $n$ and $t$ are the number of vertices and the treewidth of the input graph, respectively.
 \end{theorem}
\begin{proof}
For $h=1$ the problem can be trivially solved in linear time, so assume $h \geq 2$.
Let ${\cal D}$ be a tree decomposition of $G$ with width $t$, and let $S$ be an (unknown) optimal solution for \textsc{$I_h$-IS-Deletion} on $G$. By Lemma~\ref{bag_cover}, there are at most $h-1$ bags $X_1,X_2,\ldots,X_{h-1}$ of ${\cal D}$ such that $V(G)\setminus S\subseteq \bigcup_{i \in [h-1]}X_{i}$. Since we may assume that ${\cal D}$ has $\Ocal(n)$ nodes~\cite{Klo94}, we can enumerate the candidate sets of  bags $X_1,X_2,\ldots,X_{h-1}$ in time $\Ocal(n^{h-1})$. For each such fixed  set $X_1,X_2,\ldots,X_{h-1}$, we generate all subsets $\bar{S} \subseteq \bigcup_{i \in [h-1]}X_{i}$, which are at most $2^{(h-1)(t+1)}$ many,
and for each $\bar{S}$ we check whether the graph $G[\bar{S}]$ is $I_h$-free, in time $2^{t}\cdot t^{\Ocal(1)}\cdot n$, by computing a maximum independent set of $G[\bar{S}]$ using dynamic programming based on treewidth~\cite{CyganFKLMPPS15} (note that having treewidth at most $t$ is a hereditary property).
Note that, by Lemma~\ref{bag_cover}, there exists some of the considered sets $\bar{S}$ such that $V(G) \setminus \bar{S} = S$, and therefore an optimal solution $S$ of \textsc{$I_h$-IS-Deletion} on $G$ can be found in time $\Ocal(n^{h-1} \cdot 2^{(h-1)(t+1)} \cdot 2^{t}\cdot t^{\Ocal(1)}\cdot n) = 2^{\Ocal(t)} \cdot n^{h}$ , as claimed.
 % and, for each fixed set of bags, one can find in time $2^{\Ocal(t)} \cdot h-1$ time the maximum $I_h$-free subset of vertices. Thus, in time $2^{\Ocal(t)} \cdot n^{\Ocal(h-1)}$ time we can find $V(G)\setminus S$ (and hence $S$).
\end{proof}

We would like to mention that the approach used in the algorithm of Theorem~\ref{thm:algo-Is} does not seem to be easily applicable to the colorful version of the problem. Indeed, the colored version of Lemma~\ref{chordal_Ih-free} fails: removing a clique from a $\sigma$-$I_h$-free chordal graph does not necessarily yield a $\sigma$-$I_{h-1}$-free chordal graph, and the inductive argument does not apply.

\smallskip

To conclude this section, note that,  for any graph $H$ and any instance $(G,\sigma)$ of \textsc{Colorful $H$-IS-Deletion}, any edge between two vertices $u,v$ with $\sigma(u)=\sigma(v)$ can be safely deleted without affecting the instance. Hence, if $H=K_2$ we can assume that the input graph is bipartite, and therefore the \textsc{Colorful $K_2$-IS-Deletion} problem (where the goal is to hit all edges) is equivalent to computing a minimum vertex cover in a bipartite graph, which can be done in polynomial time. Similarly, the \textsc{Colorful $I_2$-IS-Deletion} problem can also be solved in polynomial time, by computing a minimum vertex cover in the bipartite complement of the input graph. This is in sharp contrast to the uncolored version, where both problems are \NP-hard~\cite{LeYa80}.

%\section{Algorithms for \textsc{$H$-IS-Deletion}}
%\label{sec:algos-no-colors}

%% Uéverton %%
%\input{algos-no-colors}
%\input{sketch-DP}

\section{Lower bounds for \textsc{$H$-IS-Deletion}}
\label{sec:lower-bounds-no-colors}

In this section we present lower bounds for the \textsc{$H$-IS-Deletion} problem. Our reductions will be from the 3-\textsc{Sat} problem restricted to clean formulas (see Section~\ref{sec:prelim} for the definition), and are strongly inspired by a reduction of
Cygan et al.~\cite[Theorem 4]{CyganMPP17} for the \textsc{$H$-S-Deletion} problem when $H$ is the graph obtained from $K_{2,h}$ by attaching a triangle to each of the two vertices of degree $h$. More precisely, the reduction from the 3-\textsc{Sat} problem restricted to clean formulas of Cygan et al.~\cite[Theorem 4]{CyganMPP17} is based on a frame graph that is a simplified version of the general one that we define below, but that enjoys its essential properties, namely that each occurrence of the forbidden (induced) graph $H$ corresponds to a satisfying variable/clause pair. Our technical contribution is to enhance this basic frame graph in order to deal with different forbidden subgraphs $H$, in particular (as discussed in detail below) by adding edges inside the central part, redefining the  ``attached graphs'' $L$, changing the adjacencies given by the functions $f_{C,\ell}$ defined below, or adding a new vertex set $T$ into the central part (cf. Figure~\ref{fig:frame}).

\smallskip

We start by presenting the general frame of the reductions together with some generic properties that our eventual instances of  \textsc{$H$-IS-Deletion} will satisfy, which allow to prove in a unified way (cf. Lemma~\ref{lem:properties}) the equivalence of the instances. Variations of this general frame will yield the concrete reductions for distinct graphs $H$ (cf. Theorems~\ref{thm:lower-bound-no-colors},~\ref{thm:lower-bound-no-colors-WEAKER},~\ref{thm:lower-bound-no-colors-2}, and~\ref{thm:lower-bound-complete-bipartite}).

\medskip

\noindent \textbf{General frame of the reductions}.  Given a clean  3-\textsc{Sat} formula $\varphi$ with $n$ variables and $m$ clauses, we proceed to build a so-called \emph{frame graph} $F_{H,\varphi}$. For each graph $H$ considered in the reductions, $F_{H,\varphi}$ will be enhanced with additional vertices and edges, obtaining a graph $G_{H,\varphi}$ that will be the constructed instance of the \textsc{$H$-IS-Deletion} problem.

Let $h$ be an integer depending on $H$, to be specified in each particular reduction,
 %(for instance, several times we will have $h= |V(H)| -2$),
  let $s$ be the smallest positive integer such that $s^h \geq 3n$, and note that $s = \Ocal(n^{1/h})$. We introduce a set of vertices $M = \{ w_{i,j} \mid i \in [s], j \in [h]\}$, which we call the \emph{central part} of the frame. One may think of this set $M$ as a matrix with $s$ rows and $h$ columns.  We will sometimes add an extra set $T$ of vertices to the central part, with $|T|$ depending only on $H$, obtaining an \emph{enhanced central part} $M' = M \cup T$.

 Let $L$ be a graph, to be specified according to each particular considered graph $H$. By \emph{attaching a copy of $L$ between two vertices $u,v \in V(F_{H,\varphi})$} we mean adding a new copy of $L$, choosing two arbitrary distinct vertices of $L$, and identifying them with $u$ and $v$, respectively.

For each variable $x$ of $\varphi$ and for each clause $C$ containing $x$ in a literal $\ell \in \{x, \bar{x}\}$, we add to $F_{\varphi}$ a new vertex $a_{x,C, \ell}$. We also introduce another ``dummy'' vertex $a_x$. Since $\varphi$ is clean, we have introduced four vertices in $F_{H,\varphi}$ for each variable $x$. Let $a_{x,C_1, \ell}$, $a_{x,C_2, \bar{\ell}}$, $a_{x,C_3, \ell}$, $a_{x}$ be the four introduced vertices (recall that $x$ appears at least once positively and negatively in $\varphi$). We attach a copy of $L$ between the following four pairs of vertices: $(a_{x,C_1, \ell}, a_{x,C_2, \bar{\ell}})$, $(a_{x,C_2, \bar{\ell}}, a_{x,C_3, \ell})$, $(a_{x,C_3, \ell}, a_x)$, and $(a_x, a_{x,C_1, \ell})$. We denote by $A$ the union of all the vertices in these variable gadgets.

For each clause $C$ of $\varphi$ and for each literal $\ell$ in $C$, we add to $F_{\varphi}$ a new vertex $b_{C, \ell}$. Since $\varphi$ is clean, we have introduced two or three vertices in $F_{H,\varphi}$ for each clause $C$. We attach a copy of $L$ between every pair of these vertices. We denote by $B$ the union of all the vertices in these clause gadgets. This concludes the construction of the frame $F_{H,\varphi}$; cf. Figure~\ref{fig:frame}.

\begin{figure}[ht]
\begin{center}
\includegraphics[scale=1.1]{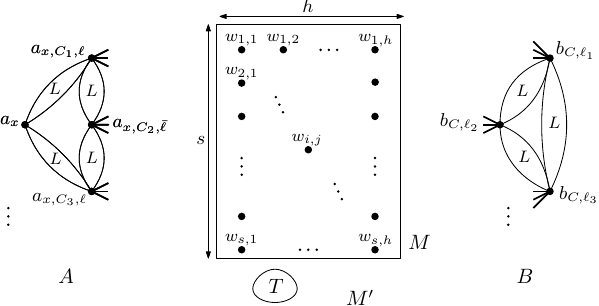}
\end{center}
\caption{Illustration of the general frame graph $F_{H,\varphi}$.}
\label{fig:frame}
\end{figure}

\noindent In all our reductions, the graph $G_{H,\varphi}$ will satisfy the following property:

 \begin{itemize}
\item[]{\sf P1}: All the connected components of $G_{H,\varphi} \setminus M'$ are of size bounded by a function of $H$.
\end{itemize}

%\noindent If Property  all the connected components of the graph $G_{H,\varphi} \setminus M$ will be of size bounded by a function of $h$, which is assumed to be constant, and therefore $\tw(G_{H,\varphi}) = \Ocal(|M|) = \Ocal(s) = \Ocal (n^{1/h})$.

%Once we will have proved that $\varphi$ is satisfiable if and only if  $G_{H,\varphi}$ is a positive instance of \textsc{$H$-IS-Deletion} for an appropriately chosen budget $k$, this bound on the treewidth of $G_{H,\varphi}$ implies that an algorithm in time $\Ocal^*(2^{o((\tw(G_{H,\varphi}))^h)}) = 2^{o(n)}$ for solving \textsc{$H$-IS-Deletion} on $G_{H,\varphi}$ would violate the \ETH.

\noindent Also, in all our reductions the budget that we set for the solution of \textsc{$H$-IS-Deletion} on $G_{H,\varphi}$ is $k:= 2n + \sum_{C \in \varphi}(|C| -1) = 5n-m$, where $|C|$ denotes the number of literals in clause $C$.  For each fixed graph $H$, the choice of $k$, the edges within $M'$, and the edges between $M'$ and the sets $A,B$ will force the following behavior in $G_{H,\varphi}$:

 \begin{itemize}
\item[]{\sf P2}: For each gadget corresponding to a variable $x$, at least one of the pairs $(a_{x,C_1,\ell},a_{x,C_3,\ell})$\\
     \textcolor{white}{{\sf P2}: }and $(a_x,a_{x,C_2,\bar{\ell}})$ needs to be in the solution and, for each gadget corresponding to a\\
     \textcolor{white}{{\sf P2}: }clause $C$, at least $|C|-1$ vertices in the set $\{b_{C,\ell} \mid \ell \in C\}$ need to be in the solution.
\end{itemize}

\noindent The above property together with the choice of $k$ imply that the budget is {\sl tight}: exactly one of the pairs $(a_{x,C_1,\ell}, a_{x,C_3,\ell})$ and $(a_x,a_{x,C_2,\bar{\ell}})$ is in the solution, thereby defining the {\sf true}/{\sf false} assignment of variable $x$; and exactly one of the vertices in $\{b_{C,\ell} \mid \ell \in C\}$ is {\sl not} in the solution, corresponding to a satisfied literal in $C$. More precisely, our graph $G_{H,\varphi}$ will satisfy the following key property:
% satisfied by the assignment of the variables.

 \begin{itemize}
\item[]{\sf P3}: Let $X \subseteq V(G_{H,\varphi})$ contain exactly one of $(a_{x,C_1,\ell}, a_{x,C_3,\ell})$ and $(a_x,a_{x,C_2,\bar{\ell}})$ for each\\
     \textcolor{white}{{\sf P3}: }variable $x$, and exactly $|C|-1$ vertices in $\{b_{C,\ell} \mid \ell \in C\}$ for each clause $C$. If $G_{H,\varphi} \setminus X$\\
     \textcolor{white}{{\sf P3}: }contains $H$  as an induced subgraph, then it has an occurrence of $H$ as an induced\\ \textcolor{white}{{\sf P3}: }subgraph containing exactly one vertex $a_{x,C,\ell} \in A$ and exactly one vertex $b_{C',\ell'} \in B$,\\ \textcolor{white}{{\sf P3}: }with $(C, \ell) = (C', \ell')$.
     %Then all\\
%      \textcolor{white}{{\sf P3}: }occurrences of $H$ in $G_{H,\varphi} \setminus X$ as an induced subgraph contain exactly one vertex\\
%       \textcolor{white}{{\sf P3}: }$a_{x,C,\ell} \in A$ and exactly one vertex $b_{C',\ell'} \in B$, with $(C, \ell) = (C', \ell')$.
       Moreover, each such a pair of vertices gives rise to an occurrence\\ \textcolor{white}{{\sf P3}: }of $H$ in $G_{H,\varphi} \setminus X$.
\end{itemize}

 %\begin{itemize}
%\item[]{\sf P2}: $\varphi$ is satisfiable if and only if $G_{H,\varphi}$ has a solution of \textsc{$H$-IS-Deletion} of size $5n-m$.
%has a solution of \textsc{$H$-IS-Deletion} containing\\
   % \textcolor{white}{vian}exactly one of the pairs $(a_{x,C_1,\ell}, a_{x,C_3,\ell})$ and $(a_x,a_{x,C_2,\bar{\ell}})$ for each variable $x$ of $\varphi$,\\
    %\textcolor{white}{vian}and all but one of the vertices in $\{b_{C,\ell} \mid \ell \in C\}$ for each clause $C$ of $\phi$.
%\end{itemize}

\noindent Note that property~{\sf P3} states that the occurrences of $H$  described above are ``representative'' of all occurrences of $H$, in the sense that it is enough that the set $X$ hits these particular occurrences in order to guarantee that $G_{H,\varphi} \setminus X$ contains no occurrence of $H$ at all. We now show that the above three properties are enough to construct the desired reductions.

\begin{lemma}\label{lem:properties}
Let $H$ be a fixed graph and, given a clean 3-{\sc Sat} formula $\varphi$, let $G_{H,\varphi}$ be a graph constructed starting from the frame graph $F_{H,\varphi}$ described above, where the central part $M$ has $h$ columns for some constant $h \geq 1$ depending on $H$. If $G_{H,\varphi}$ satisfies properties~{\sf P1}, {\sf P2}, and~{\sf P3}, then the  {\sc $H$-IS-Deletion} problem cannot be solved in time $\Ocal^*(2^{o(t^{h})})$ unless the \ETH fails, where $t$ is the width of a given tree decomposition of the input graph.
\end{lemma}
\begin{proof}
%Since $h$ depends on $H$, which is assumed to be a fixed graph, we have that $\tw(G_{H,\varphi})
Since $H$ is a fixed graph, property~{\sf P1} implies that we can easily construct in polynomial time a tree decomposition of $G_{H,\varphi}$ of width $\Ocal(|M'|) = \Ocal(|M|)= \Ocal(s) = \Ocal (n^{1/h})$. Let  $t$ be the width of the constructed tree decomposition of $G_{H,\varphi}$. We set $k:=5n-m$, where $n$ and $m$ are the number of variables and clauses of $\varphi$, respectively. We claim that $\varphi$ is satisfiable if and only if $G_{H,\varphi}$ has a solution of \textsc{$H$-IS-Deletion} of size at most $k$. This will conclude the proof of the lemma, since
an algorithm in time $\Ocal^*(2^{o(t^{h})})$ to solve \textsc{$H$-IS-Deletion} on $G_{H,\varphi}$ would imply the existence of an algorithm in time $2^{o(n)}$ to decide whether $\varphi$ is satisfiable, which would contradict the \ETH by Lemma~\ref{lem:clean3SAT}.

Suppose first that $\alpha$ is an assignment of the variables that satisfies all the clauses in $\varphi$, and we define a set $X \subseteq V(G_{H,\varphi})$ as follows. For each variable $x$, add to $X$ all vertices $a_{x,C,\ell}$ such that $\alpha(\ell)$ is {\sf true}. If only one vertex was added in the previous step, add to $X$ vertex $a_x$ as well. For each clause $X$, choose a literal $\ell$ that satisfies $C$, and add to $X$ the set $\{b_{C,\ell'} \mid \ell' \neq \ell\}$. By construction we have that $|X|=k$, and property~{\sf P3} guarantees that $H$ does not occur in $G_{H,\varphi} \setminus X$ as an induced subgraph.

Conversely, suppose that there exists $X \subseteq V(G_{H,\varphi})$ with $|X| \leq k$ such that $G_{H,\varphi} \setminus X$ does not contain $H$ as an induced subgraph. By property~{\sf P2} and the choice of $k$, $X$ contains exactly one of $(a_{x,C_1,\ell}, a_{x,C_3,\ell})$ and $(a_x,a_{x,C_2,\bar{\ell}})$ for each variable $x$, and exactly $|C|-1$ vertices in $\{b_{C,\ell} \mid \ell \in C\}$ for each clause $C$. We define the following assignment $\alpha$ of the variables: for each variable $x$, let $\ell \in \{x, \bar{x}\}$ such that $a_{x,C,\ell} \in X$ for some clause $C$. Then we set $\alpha(x)$ to {\sf true} if $\ell = x$, and to \false if $\ell = \bar{x}$. By the above discussion, this is a valid assignment. Consider a clause $C$ of $\varphi$, and let $\ell$ be the literal in $C$ such that $b_{C,\ell} \notin X$. Property~{\sf P3} and the hypothesis that $X$ is a solution imply that there exists a variable $x \in \{\ell, \bar{\ell}\}$ such that $a_{x,C,\ell} \in X$, as otherwise there would be an occurrence of $H$ in $G_{H,\varphi} \setminus X$. By the definition of $\alpha$, necessarily $\alpha(\ell)$ is {\sf true}, and therefore $\alpha$ satisfies $C$. Since this argument holds for every clause, we conclude that $\varphi$ is satisfiable.
%Once we will have proved that $\varphi$ is satisfiable if and only if  $G_{H,\varphi}$ is a positive instance of \textsc{$H$-IS-Deletion} for an appropriately chosen budget $k$, this bound on the treewidth of $G_{H,\varphi}$ implies that an algorithm in time $\Ocal^*(2^{o((\tw(G_{H,\varphi}))^h)}) = 2^{o(n)}$ for solving \textsc{$H$-IS-Deletion} on $G_{H,\varphi}$ would violate the \ETH.
\end{proof}

We now proceed to describe concrete reductions for several instantiations of $H$. In order to add edges between the enhanced central part $M'$ and the sets $A,B$, we use the following nice trick introduced in~\cite{CyganMPP17}. To each pair $(C,\ell)$, where $C$ is a clause of $\varphi$ and $\ell$ is a literal in $C$, we assign a function $f_{C,\ell}: [h] \to [s]$. Note that there are $s^h$ many such functions, and recall that $s$ has been chosen so that $s^h \geq 3n$. We assign these functions in such a way that $f_{C,\ell}\neq f_{C',\ell'}$ whenever $(C,\ell) \neq (C',\ell')$; note that this is possible by the choice of $s$ and the fact that, since $\varphi$ is clean, each clause contains at most three literals. We assume henceforth  that these functions are fixed.

We start with the following result that provides a tight lower bound for a graph that is  very ``close'' to a clique, namely a clique minus one edge.

%\ig{check whether $h=1$ also works -- YES!}

\begin{theorem}\label{thm:lower-bound-no-colors}
For any integer $h \geq 1$, the  {\sc $(K_{h+2}-e)$-IS-Deletion} problem cannot be solved in time $\Ocal^*(2^{o(t^{h})})$ unless the \ETH fails, where $t$ is the width of a given tree decomposition of the input graph.
\end{theorem}
\begin{proof}
We first treat the case $h=1$ separately, by presenting a polynomial-time reduction from the \textsc{Vertex Cover} problem, which is well-known not to be solvable, assuming the \ETH, in time $2^{o(n+m)}$ on  graphs with $n$ vertices and $m$ edges~\cite{ImpagliazzoP01-ETH,ImpagliazzoP01}. In fact, we will prove the stronger lower bound of $\Ocal^*(2^{o(n)})$, which implies the lower bound $\Ocal^*(2^{o(t)})$ corresponding to the case $h=1$ claimed in the statement of the theorem. Note that, in the case $h=1$,  $K_3 - e = P_3$. Given an instance $G$ of \textsc{Vertex Cover}, let $G'$ be obtained from $G$  by attaching a private neighbor to every vertex of $G$, and note that $|V(G')| = 2|V(G)|$. It can be easily verified that the size of a minimum vertex cover of $G$ equals the minimum size of a vertex set of $G'$ intersecting all its induced $P_3$'s. Hence, the \textsc{$(K_3 - e)$-IS-Deletion} problem cannot be solved in time $\Ocal^*(2^{o({n})})$ under the \ETH.

Suppose henceforth that $h \geq 2$, and let $H=K_{h'+2}-e$ for some $h' \geq 2$ (we relabel $h$ as $h'$ in $H$ to keep the index $h$ for the number of columns in the frame graph $F_{H,\varphi}$). We will present a reduction from the 3-\textsc{Sat} problem restricted to clean formulas. Given such a formula $\varphi$, let $F_{H,\varphi}$ be the frame graph described above the statement of the theorem,
with $h=h'$, $L = K_{h+2}-e$, and $T=\emptyset$. In this construction, when attaching copies of $L$, we choose the attachment vertices to be two distinct vertices of $L$ different from the endvertices of its unique non-edge. Note that this is always possible as $h \geq 2$. We proceed to build, starting from $F_{H,\varphi}$, an instance $G_{H,\varphi}$ of \textsc{$H$-IS-Deletion} with budget $k=5n-m$ satisfying properties~{\sf P1}, {\sf P2}, and~{\sf P3}, and then Lemma~\ref{lem:properties} will imply the claimed lower bound.

We add an edge between any two vertices $w_{i,j},w_{i',j'} \in M$ with $j \neq j'$. That is, we turn $G_{H,\varphi}[M]$ into a complete $h$-partite graph, where each part has size $s$. For each clause $C$ and each literal $\ell$ in $C$, where $\ell \in \{x,\bar{x}\}$ for some variable $x$, we add the edges $\{a_{C,x,\ell}, w_{f_{C,\ell}(j), j}\}$ and $\{b_{C,\ell}, w_{f_{C,\ell}(j), j}\}$ for every $j \in [h]$. That is, the function $f_{C,\ell}$ indicates the {\sl unique} neighbor of $a_{C,x,\ell}$ and $b_{C,\ell}$ in the $j$-th column of $M$, for every $j \in [h]$. This concludes the construction of $G_{H,\varphi}$, which clearly satisfies property~{\sf P1}. By the choice of $k$ and the fact that there is a copy of $H$ between the corresponding vertices of $A$ and $B$ (cf. Figure~\ref{fig:frame}), property~{\sf P2} holds as well. Let $X \subseteq V(G_{H,\varphi})$ be a set as in property~{\sf P3}, and let $\tilde{H}$ be an induced subgraph of $G_{H,\varphi} \setminus X$ isomorphic to $H$. Since $\omega(G_{H,\varphi}[M]) = h$, $\omega(G_{H,\varphi}[(A\cup B) \setminus X]) \leq  h$ (here we use that $h \geq 2$ and the choice of the attachment vertices of $L$),  and $\omega(H)=h+1$, $\tilde{H}$ intersects both $M$ and $A \cup B$. Moreover, since no two adjacent vertices in $(A\cup B) \setminus X$ both have neighbors in $M$, necessarily $|V(\tilde{H}) \cap (A \cup B)| =2$ and $V(\tilde{H}) \cap M$ induces a clique of size $h$, which implies that $\tilde{H}$ contains a vertex in each column of $M$. By the definition of the functions $f_{C,\ell}$ and the construction of $G_{H,\varphi}$, the two vertices in $V(\tilde{H}) \cap (A \cup B)$ must be $a_{x,C,\ell} \in A$ and $b_{C',\ell'} \in B$ with $(C, \ell) = (C', \ell')$, and therefore property~{\sf P3} follows and we are done by Lemma~\ref{lem:properties}.\end{proof}

In the following result we provide an almost tight lower bound for another  graph $H$ that is also ``close'' to a clique, in this case a clique of size $h$ plus two isolated vertices. Since this graph is somehow symmetric to the one considered in
Theorem~\ref{thm:lower-bound-no-colors}, the natural approach is to reverse the roles of neighbors and non-neighbors given by the functions $f_{C,\ell}$. However, in this way there would be many cliques of size $h$ consisting of a vertex in $A \cup B$ together with $h-1$ of its neighbors in $M$, which would create many undesired induced occurrences of $H$ with any two vertices anticomplete to such a clique. We circumvent this problem by ``reducing'' the number of columns of the central part to $h-1$, and adding a vertex $s_0$ to the set $T$ that is complete to $M$ and anticomplete to $A \cup B$. This vertex guarantees property~{\sf P3}, at the price of achieving only a near-optimal lower bound\footnote{In the conference version of this article presented at MFCS 2020, we claimed a tight bound of $\Ocal^*(2^{o(t^{h})})$ for $H=K_{h}+I_2$, but the proof contained a flaw that we have fixed in the full version.} for $H=K_{h}+I_2$, except for the case $h=1$, in which the lower bound is optimal under the \ETH. For technical reasons discussed in the proof of Theorem~\ref{thm:lower-bound-no-colors-WEAKER}, in our construction we need to assume that $h \geq 4$.

\begin{theorem}\label{thm:lower-bound-no-colors-WEAKER}
Let $h \geq 1$ be an integer. Assuming the \ETH, the {\sc $(K_{h} + I_2)$-IS-Deletion} problem cannot be solved in time $\Ocal^*(2^{o(t)})$ if $h \leq 3$, and in time $\Ocal^*(2^{o(t^{h-1})})$ if $h\geq 4$, where $t$ is the width of a given tree decomposition of the input graph.
\end{theorem}
\begin{proof}
  We first treat the cases $h \in \{1,2,3\}$ separately. As in the proof of Theorem~\ref{thm:lower-bound-no-colors}, we present polynomial-time reductions from the \textsc{Vertex Cover} problem. Again, we will prove the stronger lower bound of $\Ocal^*(2^{o(n)})$, which implies the lower bound $\Ocal^*(2^{o(t)})$ corresponding to the cases where $h \leq 3$ claimed in the statement of the theorem.

  Let first $h=1$, and note that $K_1 + I_2 = I_3$. We will show that, under the \ETH, \textsc{$K_3$-IS-Deletion} cannot be solved in time $\Ocal^*(2^{o({n})})$, which implies, by taking the complement of the input graph, that  \textsc{$I_3$-IS-Deletion} cannot be solved in time $\Ocal^*(2^{o({n})})$, concluding the proof. Given an instance $G$ of \textsc{Vertex Cover}, let $G'$ be obtained from $G$  by adding, for each edge $\{u,v\} \in E(G)$, a new vertex $w$ and two edges $\{u,w\}$ and $\{v,w\}$. Note that $|V(G')| = |V(G)| + |E(G)|$, and recall that \textsc{Vertex Cover} cannot be solved in time $2^{o(n+m)}$ under the \ETH. It can be easily verified that the size of a minimum vertex cover of $G$ equals the minimum size of a vertex set of $G'$ intersecting all its $K_3$'s. Hence, the \textsc{$K_3$-IS-Deletion} problem cannot be solved in time $\Ocal^*(2^{o({n})})$ under the \ETH.

  Let $h=2$. Given an instance $G$ of \textsc{Vertex Cover}, let $G'$ be obtained from $G$ by adding $|V(G)|$ isolated vertices. It can be easily verified that the size of a minimum vertex cover of $G$ equals the minimum size of a vertex set of $G'$ intersecting all the induced occurrences of $K_2+I_2$.

  Finally, let $h=3$. Given an instance $G$ of \textsc{Vertex Cover}, let $G'$ be obtained from $G$ by adding $|V(G)|$ isolated vertices and, for each edge $\{u,v\} \in E(G)$, a new vertex $w$ and two edges $\{u,w\}$ and $\{v,w\}$. It can be easily verified that the size of a minimum vertex cover of $G$ equals the minimum size of a vertex set of $G'$ intersecting all the induced occurrences of $K_3+I_2$.

Suppose henceforth that $h \geq 4$, and let $H=K_{h'}+I_2$ for some $h' \geq 4$ (again, we relabel $h$ as $h'$ in $H$ to keep the index $h$ for the number of columns in the frame graph $F_{H,\varphi}$). We will present a reduction from the 3-\textsc{Sat} problem restricted to clean formulas. Given such a formula $\varphi$, let $F_{H,\varphi}$ be the frame graph described above the statement of the theorem,
with $h=h'-1$, $L = K_{h'}$, and $T=\{s_0\}$ where $s_0$ is a  new vertex. %Again, when attaching copies of $L$, we choose the attachment vertices to be two distinct vertices of $L$ different from the endvertices of its unique non-edge. Note that this is always possible as $h \geq 2$.
We proceed to build, starting from $F_{H,\varphi}$, an instance $G_{H,\varphi}$ of \textsc{$H$-IS-Deletion} with budget $k=5n-m$ satisfying properties~{\sf P1}, {\sf P2}, and~{\sf P3}, and then Lemma~\ref{lem:properties} will imply the claimed lower bound.

We add an edge between any two vertices $w_{i,j},w_{i',j'} \in M$ with $j \neq j'$. That is, we turn $G_{H,\varphi}[M]$ into a complete $h$-partite graph, where each part has size $s$.
%We add an edge between any two vertices in $M$, that is, we turn $G_{H,\varphi}[M]$ into a clique.
For each clause $C$ and each literal $\ell$ in $C$, where $\ell \in \{x,\bar{x}\}$ for some variable $x$, we add the edges $\{a_{C,x,\ell}, w_{i, j}\}$ and $\{b_{C,\ell}, w_{i, j}\}$ for every  $j \in [h]$  and every  $i \in [h] \setminus \{f_{C,\ell}(j)\}$. That is, in this case, the function $f_{C,\ell}$ indicates the {\sl unique} non-neighbor of $a_{C,x,\ell}$ and $b_{C,\ell}$ in the $j$-th column of $M$, for every $j \in [h]$. We make vertex $s_0$ complete to $M$ and anticomplete to $A \cup B$.
Finally, for every copy of $L$ that we have attached in $G_{H,\varphi}$, let $v_1,\ldots,v_{h'-2}$ be the vertices of $L$ distinct from the two attachment vertices, ordered arbitrarily. For $j \in [h'-2]$, we make $v_j$ complete to the $j$-th column of $M$ (note that the last column of $M$ is not used). We add these edges for two reasons. The first one is to prevent that the non-attachment vertices of $L$ may play the role of the two isolated vertices in a potential occurrence of $K_{h'} + I_2$. The second one is to prevent that the non-attachment vertices of $L$ may participate in a clique of size $h'$ in an occurrence of $K_{h'} + I_2$. Here is where the hypothesis that $h' \geq 4$ is important. Indeed, since $h' \geq 4$, each copy of $L$ contains at least two non-attachment vertices, and the fact that each such vertex is adjacent to a {\sl distinct} column of $M$ implies that, together with one of the ``surviving'' attachment vertices and some vertices of $M$, these non-attachment vertices cannot participate in a clique of size $h'$ (for this, we also use that each column of $M$ induces an independent set).
%we make every vertex in $L$ different from the two attachment points complete to $M$. We also make each of the dummy vertices $a_x$  complete to $M$.
This concludes the construction of $G_{H,\varphi}$, which clearly satisfies property~{\sf P1}, since the vertices in the variable and clauses gadgets have neighbors only in $M$ and within those gadgets.

Let us now argue that $G_{H,\varphi}$ satisfies property~{\sf P2}. Assume for contradiction that there exists a hitting set $X$ of size at most $k$ and a variable $x$ such that none of the pairs $(a_{x,C_1,\ell},a_{x,C_3,\ell})$
and $(a_x,a_{x,C_2,\bar{\ell}})$ is entirely in $X$. The choice of $k$ and the construction of the frame graph $F_{H,\varphi}$ imply that, in that case, there exists an entire copy of $L=K_{h'}$ in
$G_{H,\varphi} \setminus X$. Since the vertices in the variable or clause gadget where this copy of $L$ lies do not have neighbors in other variable or clause gadgets, in order for a copy of $H$ not to occur in $G_{H,\varphi} \setminus X$, necessarily $X$ must contain {\sl all} vertices in {\sl all} variable and clause gadgets except possibly two of them (the one containing $L$, and another one that may allow for the occurrence of $K_{h'} + I_1$), which clearly exceeds the budget $k$. An analogous argument applies if we assume that $|X \cap \{b_{C,\ell} \mid \ell \in C\}| \leq |C|-2$ for some clause $C$. Thus, $G_{H,\varphi}$ satisfies property~{\sf P2}.

Finally, let $X$ be a set as in property~{\sf P3}, let $\tilde{H}$ be an induced occurrence of $H$ in $G_{H,\varphi} \setminus X$, and let $K$ be the subgraph of $\tilde{H}$ isomorphic to $K_h'$. The choice of $X$ and the discussion above about the edges between $M$ and the copies of $L$ imply that $K$ cannot contain a non-attachment vertex of a copy of $L$. However, $K$ may contain a variable or clause vertex $v \in A \cup B$ together with $h$ of its neighbors in $M$, one in each column. Note that, by construction of $G_{H,\varphi}$, $| V(K) \cap (A \cup B) | \leq 1$. In order to complete $K$ into $K_{h'} + I_2$, $\tilde{H}$ should contain two non-adjacent vertices in $A \cup B$ that are anticomplete to $K$. The construction of $G_{H,\varphi}$ forces that these two vertices must be  $a_{x,C,\ell} \in A$ and $b_{C',\ell'} \in B$ with $(C, \ell) = (C', \ell')$. We distinguish two cases. If $|V(K) \cap (A \cup B) |= 0$, property~{\sf P3} follows and we are done by Lemma~\ref{lem:properties}. Otherwise, $V(K) \cap (A \cup B) = \{v\}$ for some vertex $v$. Note that, since $v$ is not adjacent to $s_0$ and $s_0$ is complete to $M$, $s_0 \notin V(\tilde{H})$. We construct from $\tilde{H}$ another induced occurrence $\tilde{H}'$ of $H$  in $G_{H,\varphi} \setminus X$ by defining $V(\tilde{H}') := V(\tilde{H}) \setminus \{v\} \cup \{s_0\} $. The subgraph $\tilde{H}'$ satisfies the conditions of property~{\sf P3}, and the theorem follows.
%property~{\sf P2} and the construction of $G_{H,\varphi}$ imply that $\omega(G[A \cup B] \setminus X) < h$, and therefore
%if $\tilde{H}$ is an occurrence of $H$ in $G_{H,\varphi} \setminus$, necessarily $V(\tilde{H}) \cap M$ induces a clique $K$ of size $h$. In order to complete $K$ into $K_h + I_2$, $\tilde{H}$ should contain two non-adjacent vertices in $A \cup B$ that are anticomplete to $K$. The construction of $G_{H,\varphi}$ forces that these two vertices must be  $a_{x,C,\ell} \in A$ and $b_{C',\ell'} \in B$ with $(C, \ell) = (C', \ell')$, and therefore property~{\sf P3} follows and we are done by Lemma~\ref{lem:properties}.
\end{proof}

Note that, in the proof of Theorem~\ref{thm:lower-bound-no-colors} for $H=K_{h+2}-e$, all the occurrences of $H$ in $G_{H,\varphi}$ are induced, and therefore the lower bound also applies to the \textsc{$(K_{h+2} - e)$-S-Deletion} problem. On the other hand, for $H=K_{h} + I_2$ the proof of Theorem~\ref{thm:lower-bound-no-colors-WEAKER} strongly uses the fact that $H$ cannot occur as an {\sl induced} subgraph. The following lemma explains why the proof does not work for the subgraph version: it can be easily solved in single-exponential time. This points out an interesting difference between both problems.

\begin{lemma}\label{thm:difference-non-induced}
 For every two integers $h \geq 1$ and $\ell \geq 0$,  the  {\sc $(K_h + I_{\ell})$-S-Deletion} problem can be solved in time $\Ocal^*(2^{\Ocal(t)})$, where $t$ is the width of a given tree decomposition of the input graph.
\end{lemma}
\begin{proof}
We will use the fact that, as observed in~\cite{CyganMPP17},  \textsc{$K_h$-S-Deletion} can be solved in time $\Ocal^*(2^{\Ocal(t)})$ for every $h \geq 1$. This proves the result for $\ell = 0$, so let now $\ell \geq 1$.
Without loss of generality, assume that $n\geq h+\ell$, as otherwise the solution is zero.
Given an $n$-vertex input graph $G$ together with a tree decomposition of width $t$, we first solve \textsc{$K_h$-S-Deletion} on $G$ in time  $\Ocal^*(2^{\Ocal(t)})$. Let $X$ be a smallest $K_h$-hitting set in $G$ and let $|X|=k$.
Notice that if $n\geq h+\ell$, whenever $G$ contains a clique of size $h$ then it contains $K_h + I_{\ell}$ as a subgraph as well.
Thus, if $k \geq n - (h+ \ell) + 1$, then $|V(G) \setminus X| < h + \ell$, so we can safely output $n - (h+\ell) + 1$ as the size of a smallest $(K_h + I_{\ell})$-hitting set in $G$. Otherwise, if $k \leq n - (h+ \ell)$, we output $k$. The algorithm is correct by the fact that a $(K_h + I_{\ell})$-hitting set is not smaller than a $K_h$-hitting set.
\end{proof}

By Theorem~\ref{thm:generic-algo-no-colors}, the lower bound presented in Theorem~\ref{thm:lower-bound-no-colors} for $H = K_{h+2}-e$ is tight under the \ETH, and the one presented in Theorem~\ref{thm:lower-bound-no-colors-WEAKER} for $H =K_{h} + I_2$ is almost tight. These two graphs are very symmetric, in the sense that each of them contains two non-adjacent vertices that are either complete or anticomplete to a ``central'' clique $K_h$ (cf. Figure~\ref{fig:graphs-hard}). Unfortunately, for graphs without two such non-adjacent symmetric vertices, our framework described above is not capable of obtaining (almost) tight lower bounds. For instance, let $H=K_{h+1} + I_1$, that is, a clique of size $h+1$ plus an isolated vertex. Let $a \in V(H)$ be any vertex in $K_{h+1}$ and let $b$ be the isolated vertex. The natural idea in order to obtain a tight lower bound of $\Ocal^*(2^{o(t^{h})})$ would be, starting from the frame graph $F_{H,\varphi}$ described above, to make the vertices $a_{x,C,\ell} \in A$ play the role of $a$, and vertices $b_{C,\ell} \in B$ that of $b$. Then, the functions $f_{C,\ell}$ would indicate, for each vertex $a_{x,C,\ell}$ (resp. $b_{C,\ell}$), its unique neighbor (resp. unique non-neighbor) in each of the columns of $M$. However, this idea does not work for the following reason: many undesired copies of $H$ appear by interchanging the expected  roles of vertices $a_{x,C,\ell}$ and $b_{C,\ell}$, and selecting, in each column of $M$, any of the $s-1$ non-neighbors of $a$ and any of the $s-1$ neighbors of $b$.  We overcome this problem by ``pledging'' one column of $M$ and introducing a \emph{sentinel} vertex $s_0 \in T$ that is complete to $M$ and the vertices $a_{x,C,\ell}$, and anticomplete to the vertices $b_{C,\ell}$. This vertex ``fixes'' the roles of vertices in $A$ and $B$ at the price of losing one column of $M$, hence getting (by Lemma~\ref{lem:properties}) a weaker lower bound of $\Ocal^*(2^{o(t^{h-1})})$, as in Theorem~\ref{thm:lower-bound-no-colors-WEAKER}. In the following theorem we formalize this idea and we extend it to a more general graph $H$, namely $K_{h+1} + v_x$ for $0 \leq x \leq h-1$, defined as the graph obtained from $K_{h+1}$ by adding a vertex $v$ adjacent to $x$ vertices in the clique (cf. Figure~\ref{fig:graphs-hard}). We will need an extra sentinel vertex in $T$ for each of the neighbors of $v$ in the clique, losing one column of the central part $M$ for each of them.

%The graphs $H$ considered in Theorem~\ref{thm:lower-bound-no-colors} and Theorem~\ref{thm:lower-bound-no-colors-2} are depicted in Figure~\ref{fig:graphs-hard}.

\begin{figure}[ht]
\begin{center}
\includegraphics[scale=.95]{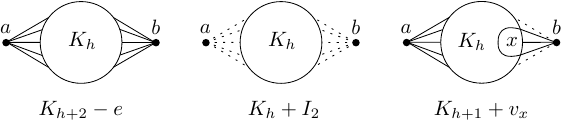}
\end{center}
\caption{Graphs $H$ considered in Theorem~\ref{thm:lower-bound-no-colors}, Theorem~\ref{thm:lower-bound-no-colors-WEAKER}, and Theorem~\ref{thm:lower-bound-no-colors-2}, respectively.}
\label{fig:graphs-hard}
\end{figure}

\begin{theorem}\label{thm:lower-bound-no-colors-2}
Let $h \geq 1$ and $0 \leq x \leq h-1$ be  integers and let $K_{h+1} + v_x$ be the graph obtained from $K_{h+1}$ by adding a vertex adjacent to $x$ vertices in the clique. Then, unless the \ETH fails, the {\sc $(K_{h+1} + v_x)$-IS-Deletion} problem cannot be solved in time $\Ocal^*(2^{o(t^{h-x-1})})$, where $t$ is the width of a given tree decomposition of the input graph.
\end{theorem}
\begin{proof}
Let $H = K_{h'+1} + v_x$ for $h' \geq 1$ and $0 \leq x \leq h'-1$ (again, we relabel $h$ as $h'$ in $H$ to keep the index $h$ for the number of columns in the frame graph $F_{H,\varphi}$).
We will again  present a reduction from the 3-\textsc{Sat} problem restricted to clean formulas. Given such a formula $\varphi$, let $F_{H,\varphi}$ be the frame graph described above with $h = h' - x - 1$, $T = \{s_0,s_1, \ldots, s_x\}$, and  $L = K_{h'+1}$ if $x=0$ and $L = H$ if $x \geq 1$.
We proceed to build, starting from $F_{H,\varphi}$, an instance $G_{H,\varphi}$ of \textsc{$H$-IS-Deletion} with budget $k=5n-m$ satisfying properties~{\sf P1}, {\sf P2}, and~{\sf P3}.

We first add an edge between any two vertices $w_{i,j},w_{i',j'} \in M$ with $j \neq j'$. That is, we turn $G_{H,\varphi}[M]$ into a complete $h$-partite graph, where each part has size $s$. We make $s_0 \in T$  complete to $A \cup M$ and anticomplete to $B$, and all vertices in $T \setminus \{s_0\}$ complete to $A \cup B \cup M$. We also turn $G_{H,\varphi}[T]$ into a clique.  For each clause $C$ and each literal $\ell$ in $C$, where $\ell \in \{x,\bar{x}\}$ for some variable $x$, we add the edges $\{a_{C,x,\ell}, w_{f_{C,\ell}(j), j}\}$ and $\{b_{C,\ell}, w_{i, j}\}$ for every $j \in [h]$ and every $i \in [h] \setminus \{f_{C,\ell}(j)\}$. That is, the functions $f_{C,\ell}$ indicate the neighbors of the vertices in $A$ and the non-neighbors of the vertices in $B$.

Moreover, in the case $x=0$ (that is, when $H=K_{h'+1} + I_1$, which is disconnected), for each copy of $L$ in $G_{H,\varphi}[B]$, we do the following. Let $v_1,\ldots,v_{h'-1}$ be the vertices of $L$ distinct from the two attachment vertices, ordered arbitrarily. For $j \in [h'-1]$, we make $v_j$ complete to the $j$-th column of $M$. We add these edges to prevent that the non-attachment vertices of $L$ may play the role of the isolated vertex in a potential occurrence of $K_{h'+1} + I_1$.

This concludes the construction of $G_{H,\varphi}$, which can be easily seen to satisfy properties~{\sf P1} and~{\sf P2} using arguments analogous to those used in the proof of Theorem~\ref{thm:lower-bound-no-colors}. Let $X \subseteq V(G_{H,\varphi})$ be a set as in property~{\sf P3}, and let $\tilde{H}$ be an induced subgraph of $G_{H,\varphi} \setminus X$ isomorphic to $H$.

Suppose first that $x \geq 1$. Since $\omega(G_{H,\varphi}[M \cup T]) = h'$, $\omega(G_{H,\varphi}[(A\cup B) \setminus X]) \leq  h'$,  and $\omega(H)=h'+1$, $\tilde{H}$ intersects both $M'= M \cup T$ and $A \cup B$. Moreover, since no two adjacent vertices in $(A\cup B) \setminus X$ both have neighbors in $M$, necessarily $|V(\tilde{H}) \cap (A \cup B)| =2$ and $V(\tilde{H}) \cap M'$ induces a clique of size $h'$, which implies that $\tilde{H}$ contains a vertex in each column of $M'$ and the whole set $T$. By the definition of the functions $f_{C,\ell}$ and the construction of $G_{H,\varphi}$, the two vertices in $V(\tilde{H}) \cap (A \cup B)$ must be $a_{x,C,\ell} \in A$ and $b_{C',\ell'} \in B$ with $(C, \ell) = (C', \ell')$, and therefore property~{\sf P3} follows and we are done by Lemma~\ref{lem:properties}.

Finally, when $x=0$, the same arguments yield that $V(\tilde{H}) \cap M'$ induces a clique $K$ of size $h'$. Note that $K$ must contain one vertex from each column of $M$ and the whole set $T$. By construction of $G_{H,\varphi}$, the only vertices in $A \cup B$ that can be complete to $K$ (resp. complete to $T \setminus \{s_0\}$ and anticomplete to $K \setminus (T \setminus \{s_0\})$) are the vertices $a_{x,C,\ell}$ (resp. $b_{C,\ell}$). Hence, the two vertices in $V(\tilde{H}) \cap (A \cup B)$ must be $a_{x,C,\ell} \in A$ and $b_{C',\ell'} \in B$ with $(C, \ell) = (C', \ell')$, and therefore property~{\sf P3} also follows and we are done again by Lemma~\ref{lem:properties}.
\end{proof}

It is worth mentioning that the lower bound given in Theorem~\ref{thm:lower-bound-no-colors-2} can be strengthened to $\Ocal^*(2^{o(t^{h-\tilde{x}-1})})$, where $\tilde{x}= \min\{x, h-1-x\}$, by using the following trick. If the vertex not belonging to the clique $K_{h+1}$ (vertex $b$ in Figure~\ref{fig:graphs-hard}) has more than $\frac{h-1}{2}$ neighbors in the clique (i.e., if $\tilde{x}= h-1-x$), we can interchange, for vertices $b_{C,\ell} \in B$, the roles of neighbors/non-neighbors of the set $T \setminus \{s_0\}$ and the vertices in $M$ given by the functions $f_{C,\ell}$. Doing this modification, an analogous proof works and we can keep the number of columns of $M$ to be $h-\tilde{x}-1$, which is always at least $\frac{h-3}{2}$.  We omit the details.

Another direction for transferring the lower bounds of Theorem~\ref{thm:lower-bound-no-colors} and Theorem~\ref{thm:lower-bound-no-colors-2} to other graphs $H$ is to consider complete bipartite graphs.

\begin{theorem}\label{thm:lower-bound-complete-bipartite}
For any integer $h \geq 2$, the {\sc $K_{h,h}$-IS-Deletion} problem cannot be solved in time $\Ocal^*(2^{o(t^{h})})$ unless  the \ETH fails, where $t$ is the width of a given tree decomposition of the input graph.
\end{theorem}
\begin{proof}
Let $H = K_{h',h'}$ for $h' \geq 2$. Given a clean 3-\textsc{Sat} formula $\varphi$, let $F_{H,\varphi}$ be the frame graph described above with $h = h'$, $T =\emptyset$, and  $L = K_{h,h}$. In this reduction we will slightly change the budget and property~{\sf P3} of the constructed instance $G_{H,\varphi}$ of \textsc{$H$-IS-Deletion}.
% with budget $k=5n-m$ satisfying properties~{\sf P1}, {\sf P2}, and~{\sf P3}.

Starting from $F_{H,\varphi}$,
%we add an edge between any two vertices $w_{i,j},w_{i',j'} \in M$ with $j = j'$. That is, we turn each column of $[M]$ into a clique of size $s$.
for each clause $C$ and each literal $\ell$ in $C$, where $\ell \in \{x,\bar{x}\}$ for some variable $x$, we add the edges $\{a_{C,x,\ell}, w_{f_{C,\ell}(j), j}\}$ and $\{b_{C,\ell}, w_{f_{C,\ell}(j), j}\}$ for every $j \in [h]$. We duplicate $h-2$ times the subgraph $G_{H,\varphi}[A]$, obtaining $h-1$ copies overall, and each copy has the same neighborhood in $G_{H,\varphi}[M \cup B]$ as the original one. We relabel the vertices $a_{x,C,\ell}$ in each copy as $a^{\beta}_{x,C,\ell}$ for $\beta \in [h-1]$, and we call again $A$ the set $V(G_{H,\varphi}) \setminus (M \cup B)$.  This concludes the construction of $G_{H,\varphi}$, which clearly satisfies properties~{\sf P1} and~{\sf P2}, where the latter one applies to all the pairs $(a^{\beta}_{x,C_1,\ell}, a^{\beta}_{x,C_3,\ell})$ and $(a^{\beta}_x,a^{\beta}_{x,C_2,\bar{\ell}})$ for $\beta \in [h-1]$. Hence, we update the budget accordingly to $k:= 5n-m+2(h-2)n = (2h+1)n - m$.

We proceed to prove the following modified version of property~{\sf P3} adapted to the current construction:

 \begin{itemize}
\item[]${\sf P3}'$: Let $X \subseteq V(G_{H,\varphi})$ contain exactly one of $(a^{\beta}_{x,C_1,\ell}, a^{\beta}_{x,C_3,\ell})$ and $(a^{\beta}_x,a^{\beta}_{x,C_2,\bar{\ell}})$ for each\\
     \textcolor{white}{{\sf P3'}: }variable $x$ and $\beta \in [h-1]$, and $|C|-1$ vertices in $\{b_{C,\ell} \mid \ell \in C\}$ for each clause $C$.\\
      \textcolor{white}{{\sf P3'}: }Then all occurrences of $K_{h,h}$ in $G_{H,\varphi} \setminus X$ as an induced subgraph contain a set\\
       \textcolor{white}{{\sf P3'}: }$\{ a^{\beta}_{x,C,\ell}  \mid \beta \in [h-1]\} \subseteq A$, and exactly one vertex $b_{C',\ell'} \in B$, with $(C, \ell) = (C', \ell')$.\\
       \textcolor{white}{{\sf P3'}: }Moreover, each such a vertex set gives rise to an occurrence of $H$ in $G_{H,\varphi} \setminus X$.
\end{itemize}

\noindent Basically, property~${\sf P3}'$ states that all the copies of the former set $A$ behave in a similar way. With this in mind, it is easy to see that  Lemma~\ref{lem:properties} still holds if we replace, for an instance $G_{H,\varphi}$ constructed in this reduction,  property~${\sf P3}$ by property~${\sf P3}'$.

Consider a set $X \subseteq V(G_{H,\varphi})$ as in property~${\sf P3}'$, and let $\tilde{H}$ be an induced subgraph of $G_{H,\varphi} \setminus X$ isomorphic to $K_{h,h}$. By  construction of $G_{H,\varphi}$, one of the two parts of $\tilde{H}$ must lie entirely inside $M$, as there are no edges among distinct variable or clause gadgets. Since the other part of $\tilde{H}$ must lie entirely inside $A \cup B$, the choice of the functions $f_{C,\ell}$ implies that the only vertex sets of size $h$ in $A \cup B$ that are complete to a set of $h$ non-adjacent vertices are of the form $\{ a^{\beta}_{x,C,\ell}  \mid \beta \in [h-1]\} \cup \{b_{C,\ell}\}$ for some clause $C$ and a literal $\ell \in \{x, \bar{x}\}$, hence property~${\sf P3}'$ holds and the theorem follows.
\end{proof}

Note that the proof of Theorem~\ref{thm:lower-bound-complete-bipartite} works for both \textsc{$K_{h,h}$-IS-Deletion} and \textsc{$K_{h,h}$-S-Deletion}, since all the occurrences of $K_{h,h}$ in the constructed graph $G_{H,\varphi}$ are induced.  Hence, as the particular case of Theorem~\ref{thm:lower-bound-complete-bipartite} for $h=2$ we get the following corollary, which answers a question of Mi. Pilipczuk~\cite{Pilipczuk11} about the asymptotic complexity of \textsc{$C_4$-S-Deletion} parameterized by treewidth.

\begin{corollary}\label{cor:C4}
Neither  {\sc $C_4$-IS-Deletion} nor {\sc $C_4$-S-Deletion} can be solved in time $\Ocal^*(2^{o(t^{2})})$ unless  the \ETH fails, where $t$ is the width of a given tree decomposition of the input graph.
\end{corollary}

As mentioned in~\cite{Pilipczuk11}, \textsc{$C_4$-S-Deletion} can be easily solved in time $\Ocal^*(2^{\Ocal(t^{2})})$. This fact together with Theorem~\ref{thm:generic-algo-no-colors} imply that both lower bounds of Corollary~\ref{cor:C4} are tight.% under the \ETH.

We can obtain lower bounds for other graphs $H$ that are ``close'' to a complete bipartite graph. Indeed, note that the lower bound of Theorem~\ref{thm:lower-bound-complete-bipartite} also applies to the graph $H$ obtained from $K_{h,h}$ by turning one of the two parts into a clique: the same reduction works similarly, and the only change in the construction is to turn the whole central part $M$ into a clique. We can also consider complete bipartite graphs $K_{a,b}$ with parts of different sizes, by letting the number of columns of the central part $M$ be equal to $\max\{a,b\}$, hence obtaining a lower bound of $\Ocal^*(2^{o(t^{\max\{a,b\}})})$. Similarly, we can also turn one of the two parts of
$K_{a,b}$ into a clique, and obtain the same lower bound. In particular, in this way we can obtain a lower bound of $\Ocal^*(2^{o(t^{h})})$ for the graph $H$ obtained from $K_{h+3}$ by removing the edges in a triangle.

\section{Lower bounds for \textsc{Colorful $H$-IS-Deletion}}
\label{sec:with-colors}

Our main reduction for the colored version is again strongly inspired by the corresponding reduction of Cygan et al.~\cite[Theorem 2]{CyganMPP17} for the non-induced version, again based on a reduction from the 3-\textsc{Sat} problem restricted to clean formulas and the frame graph defined in Section~\ref{sec:lower-bounds-no-colors}. The main difference with respect to their reduction is that in the non-induced version, the graph $H$ is required to contain a connected component that is neither a {\sl clique} nor a {\sl path}, while for the induced version we only require a component that is not a {\sl clique}, and therefore we need extra arguments to deal with the case where all the connected components of $H$ are paths. In particular, in the proof of Theorem~\ref{thm:lower-bound-colors}, the definition of the graphs $L_A$ and $L_B$ where the graph $H_0$ is a path is inspired from the proof of~\cite[Theorem 22]{CyganMPP17}.

\begin{theorem}\label{thm:lower-bound-colors}
Let $H$ be a graph having a connected component on $h$ vertices that is not a clique. Then {\sc Colorful $H$-IS-Deletion} cannot be solved in time $\Ocal^*(2^{o(t^{h-2})})$ unless the \ETH fails, where $t$ is the width of a given tree decomposition of the input graph.
\end{theorem}
\begin{proof}
Let $H_0, H_1, \ldots, H_p$ be the connected components of $H$, where $H_0$ is not a clique. Hence, $|V(H_0)| \geq 3$. As in Section~\ref{sec:lower-bounds-no-colors}, we will again reduce from the 3-\textsc{Sat} problem restricted to clean formulas. Given such a formula $\varphi$ with $n$ variables and $m$ clauses, we proceed to construct an instance $(G_{H,\varphi}, \sigma)$ of \textsc{Colorful $H$-IS-Deletion} such that $\varphi$ is satisfiable if and only if $G$ has a set $X \subseteq V(G)$ of size at most $k:= 15n -4m$ hitting all induced $\sigma$-$H$-subgraphs of $G$, satisfying properties~{\sf P1}, {\sf P2}, and~{\sf P3} (the latter one, concerning only  colorful copies of $H$, of course), and  then Lemma~\ref{lem:properties} will imply the claimed lower bound. The choice of the budget of the current reduction will become clear below, and does not affect the main properties of the reduction.

We start with the frame graph $F_{H,\varphi}$ defined in Section~\ref{sec:lower-bounds-no-colors},  with $h=|V(H_0)| -2 \geq 1$, $T=\emptyset$, and $L$ to be specified later. Let the vertices of $H_0$ be labeled $z_0,z_1, \ldots,z_{h+1}$ such that $\{z_0,z_{h+1}\} \notin E(H_0)$; note that this is always possible since $H_0$ is not a clique. We define the $H$-coloring $\sigma$ of $G_{H,\varphi}$ starting from the vertices of $F_{H,\varphi}$ except for the non-attachment vertices of the graphs $L$, whose coloring will be defined later together with the description of $L$. Namely, for each variable $x$ and each clause $C$ containing $\ell \in \{x, \bar{x}\}$, we define $\sigma(a_{x,C,\ell}) = \sigma(a_x) = z_0$ and $\sigma(b_{C,\ell})=z_{h+1}$. For every $i \in [s]$ and $j \in [h]$ we define $\sigma(w_{i,j}) = z_j$. That is, the vertices $a_{x,C,\ell}$ (resp. $b_{C,\ell}$) are mapped to $z_0$ (resp. $z_{h+1}$), and the whole $j$-th column of $M$ is mapped to $z_j$ for $j \in [h]$. We now add edges among the already colored vertices as follows. Within $M$, the edges mimic those in $H_0$: for any two vertices $w_{i,j}, w_{i',j'} \in M$ we add the edge $\{w_{i,j}, w_{i',j'}\}$ in $G_{H,\varphi}$ if and only if $\{z_j,z_{j'}\} \in E(H_0)$. As for the edges between $A \cup B$ and $M$, the functions $f_{C,\ell}$ capture the existence or non-existence of the edges in $H_0$ between the corresponding vertices. Namely, for each clause $C$ and each literal $\ell$ in $C$, where $\ell \in \{x,\bar{x}\}$ for some variable $x$, and for every $j \in [h]$, we do the following:

\begin{itemize}
\item If $\{z_0,z_j\}\in E(H_0)$, then we add the edge $\{a_{C,x,\ell}, w_{f_{C,\ell}(j), j}\}$. Otherwise (i.e., if $\{z_0,z_j\}\notin E(H_0)$), we add the edges $\{a_{C,x,\ell}, w_{i, j}\}$ for every $i \in [h] \setminus \{f_{C,\ell}(j)\}$.
\item  Similarly, if $\{z_{h+1},z_j\}\in E(H_0)$, then we add the edge  $\{b_{C,\ell}, w_{f_{C,\ell}(j), j}\}$. Otherwise (i.e., if $\{z_{h+1},z_j\}\notin E(H_0)$), we add the edges $\{b_{C,\ell}, w_{i, j}\}$ for every $i \in [h] \setminus \{f_{C,\ell}(j)\}$.
\end{itemize}

We now proceed to describe the graph $L$ together with its $H$-coloring. In fact, we define different (but very similar) graphs $L$ for the copies to be attached in $A$ and $B$; we call them $L_A$ and $L_B$, respectively. We start with the definition of $L_A$, and we distinguish two cases according to $H_0$:

\begin{itemize}
\item Suppose first that $H_0$ is not a path. As observed in~\cite{CyganMPP17}, it is easy to see that every graph that is not a path contains at least three vertices that are not separators. Let $z_{\beta}, z_{\gamma}$ be two such non-separating vertices of $H_0$ different from $z_0$. We define $L_A$ as the graph obtained from the disjoint union of three copies $H^1_0,H^2_0,H^3_0$ of $H_0$ by identifying the vertices $z_{\beta}$ of $H^1_0$ and $H^2_0$, and the vertices $z_{\gamma}$ of $H^2_0$ and $H^3_0$. See Figure~\ref{fig:ExampleLA}(a) for an example.

\item Suppose now that $H_0$ is a path. Let $z_{\beta}$ be an endvertex of $H_0$ different from $z_0$, and let $z_{\gamma}$ be an internal vertex of $H_0$ different from $z_0$. Note that the latter choice is always possible as $|V(H_0)| \geq 3$ and, if $z_0$ were the only internal vertex of $H_0$, then $z_0$ would not have a non-neighbor in $H_0$, contradicting the choice of $z_0$ and $z_{h+1}$. We define $L_A$ as the graph obtained from the disjoint union of three copies $H^1_0,H^2_0,H^3_0$ of $H_0$ by identifying the vertices $z_{\beta}$ of $H^1_0$ and $H^2_0$, identifying the vertices $z_{\gamma}$ of $H^2_0$ and $H^3_0$, and adding an edge between every vertex in $V(H_0^2) \setminus \{z_{\gamma}\}$ and every vertex in $V(H_0^3) \setminus \{z_{\gamma}\}$. See Figure~\ref{fig:ExampleLA}(b) for an example.
\end{itemize}

\begin{figure}[ht]
\begin{center}
\includegraphics[scale=.9]{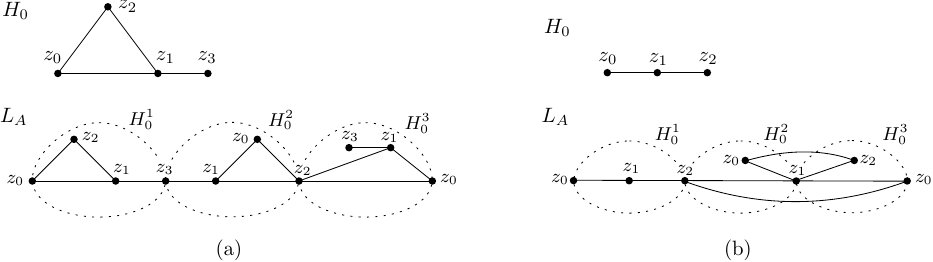}
\end{center}
\caption{Examples of the construction of $L_A$: (a) A graph $H_0$ that is not a path, with $z_{\beta}=z_3$ and $z_{\gamma}=z_2$. (b) A path $H_0$, with $z_{\beta}=z_2$ and $z_{\gamma}=z_1$. For better visibility, the edges $\{z_0,z_0\}$ and $\{z_2,z_2\}$ between $H_0^2$ and $H_0^3$ in $L_A$ are not shown in the figure.}
\label{fig:ExampleLA}
\end{figure}

In both the above cases, when we attach a copy of $L_A$ between two vertices $u,v \in A$, we identify $u$ and $v$ with the vertices $z_0$ of the first and third copies of $H_0$, respectively. The $H$-coloring $\sigma$ of $L_A$ is defined naturally, that is, each vertex gets its original color in $H_0$.

The graph $L_B$ is defined in a completely symmetric way, just by replacing vertex $z_0 \in V(H_0)$ in the above definition of $L_A$ by vertex $z_{h+1} \in V(H_0)$.

\begin{claim}\label{cl:VCbehavior}
In the graph $L_A$ (resp. $L_B$) defined above, where $u$ and $v$ are the attachment vertices, there are exactly two vertex sets $X_1,X_2 \subseteq V(L_A)\ (\text{resp. } V(L_B))$  of minimum size hitting all induced $\sigma$-$H_0$-subgraphs of $L_A$ (resp. $L_B$), $|X_1| = |X_2| =2$, $X_1 \cap \{u,v\} = \{u\}$, and $X_2 \cap \{u,v\} = \{v\}$.
\end{claim}
\begin{proof} By symmetry, it suffices to present the proof for $L_A$. In order to prove the claim, it is enough to prove that there are exactly three induced $\sigma$-$H_0$-subgraphs in $L_A$, corresponding to the three copies $H^1_0,H^2_0,H^3_0$ of $H_0$. Indeed, once this is proved, there are exactly two minimum-sized hitting sets in $L_A$: $u$ together with vertex $z_{\gamma} \in V(H_0^2) \cap V(H_0^3)$, and $v$ together with vertex $z_{\beta} \in V(H_0^1) \cap V(H_0^2)$.

%We call a vertex in $L_A$ \emph{internal} if it is different from $u$ and $v$, and different from the two vertices $z_{\beta}$ and $z_{\gamma}$ identified in $H^1_0,H^2_0$ and in $H^2_0,H_0^3$, respectively.
So suppose for contradiction that there exists an induced $\sigma$-$H_0$-subgraph $\tilde{H_0}$ in $L_A$ containing vertices in $V(H_0^i) \setminus V(H_0^j)$ and $V(H_0^j) \setminus V(H_0^i)$ for some distinct $i,j \in [3]$. (For notational simplicity, we interpret an induced $\sigma$-$H_0$-subgraph as an induced subgraph of $L_A$ isomorphic to $H_0$, with matching colors.) We distinguish two cases depending on $H_0$.

If $H_0$ is not a path, the existence of such an $\tilde{H_0}$ in $L_A$ would imply, by the construction of $L_A$ and since $H_0$ is connected, that at least one of $z_{\beta}$ and $z_{\gamma}$ is a separator in $H_0$, contradicting their choice.

Otherwise, if $H_0$ is a path, first note that since the vertex $z_{\beta} \in V(H_0^1) \cap V(H_0^2)$ is an endvertex of $H_0$, $\tilde{H_0}$ cannot contain vertices in $V(H_0^1) \setminus V(H_0^2)$ and $V(H_0^2) \setminus V(H_0^1)$. Hence, necessarily $\tilde{H_0}$ contains vertices in both $V(H_0^2) \setminus V(H_0^3)$ and $V(H_0^3) \setminus V(H_0^2)$. In particular, note that $z_{\gamma} = V(H_0^2) \cap V(H_0^3) \in V(\tilde{H_0})$. Since in $L_A$, vertex $z_{\gamma}$ was chosen as an internal vertex of $H_0$, let $z_i$ and $z_j$ be the two neighbors of $z_{\gamma}$ in $\tilde{H_0}$. Recall that in $L_A$ we added all edges between $V(H_0^2) \setminus \{z_{\gamma}\}$ and $V(H_0^3) \setminus \{z_{\gamma}\}$. We distinguish two cases:
\begin{itemize}
\item Suppose first that both $z_i,z_j \in V(H_0^2) \setminus V(H_0^3)$ or $z_i, z_j \in V(H_0^3) \setminus V(H_0^2)$. Assume that the former case holds --the other one being symmetric-- and let $z_{r}$ be a vertex of $\tilde{H_0}$ in $V(H_0^3) \setminus V(H_0^2)$, which exists by hypothesis. Then $\{z_{\gamma}, z_i, z_j,z_r\}$  induces a $C_4$ in  $\tilde{H_0}$, a contradiction since $\tilde{H_0}$ is a path.
\item Otherwise, suppose without loss of generality that $z_i \in V(H_0^2) \cap V(H_0^3)$ and $z_j \in V(H_0^3) \cap V(H_0^2)$. Then $\{z_{\gamma}, z_i, z_j\}$  induces a $C_3$ in  $\tilde{H_0}$ (for example, in Figure~\ref{fig:ExampleLA}(b), the vertices $\{z_1,z_0,z_2\}$ with $z_1 \in V(H_0^2) \cap V(H_0^3)$ induce a $C_3$), a contradiction again.\vspace{-.4cm}
\end{itemize}\end{proof}

%Hence, it satisfies Claim~\ref{cl:VCbehavior} as well.

Note that Claim~\ref{cl:VCbehavior} justifies  the budget $k$ of our eventual instance $(G_{H,\varphi}, \sigma)$ of \textsc{Colorful $H$-IS-Deletion}: any optimal solution $X$ needs to contain one of the pairs
$(a_{x,C_1,\ell}, a_{x,C_3,\ell})$ and $(a_x,a_{x,C_2,\bar{\ell}})$ for each variable $x$, and $|C|-1$ vertices in $\{b_{C,\ell} \mid \ell \in C\}$ for each clause $C$. Moreover, $X$ contains an extra internal vertex for each of the gadgets $L_A$ and $L_B$. Taking into account that the number of clauses in $\varphi$ with exactly three (resp. two) literals equals $3n-2m$ (resp. $3m-3n$), the number of gadgets $L_A$ or $L_B$ equals $4n + (3m-3n) + 3(3n-2m)$. Therefore, this amounts to a budget of
$$
2n + \sum_{C \in \varphi}(|C|-1) + 4n + (3m-3n) + 3(3n-2m)=  15n-4m = k.
$$
 % is twice the budget of the reductions presented in Section~\ref{sec:lower-bounds-no-colors}.
Finally, for every $i \in [p]$, add $k+1$ disjoint copies of $H_i$, and color their vertices according to their colors in $H$. This concludes the construction of $(G_{H,\varphi},\sigma)$,
which clearly satisfies property~{\sf P1}. Note that since for every  $i \in [p]$ the number of copies of $H_i$ in $G_{H,\varphi}$ exceeds the budget, hitting all colorful induced copies of $H$ in $G_{H,\varphi}$ with at most $k$ vertices is equivalent to hitting all colorful induced copies of component $H_0$. Therefore, Claim~\ref{cl:VCbehavior} implies that  $(G_{H,\varphi},\sigma)$ satisfies property~{\sf P2} as well.

Consider a set $X \subseteq V(G_{H,\varphi})$  as in property~{\sf P3}, and let $\tilde{H}$ be an induced $\sigma$-$H$-subgraph of $G_{H,\varphi} \setminus X$. Since for every $j \in [h]$ the only vertices of $G_{H,\varphi}$ colored $z_j$ by $\sigma$ are those in the $j$-th column of $M$, we conclude that, for every $j \in [h]$,  $\tilde{H}$ contains exactly one vertex from the $j$-column of $M$. Since $H_0$ is connected and the only vertices in $G_{H,\varphi}$ colored $z_0$ (resp. $z_{h+1}$) with neighbors in $M$ are those of type $a_{x,C,\ell}$ (resp. $b_{C,\ell}$), it follows that $\tilde{H}$ contains exactly one vertex $a_{x,C,\ell}$ and exactly one vertex $b_{C',\ell'}$.
The properties of the functions $f_{C,\ell}$, which define the edges between $A \cup B$ and $M$, and the fact that $\tilde{H}$ needs to be an {\sl induced} $\sigma$-$H$-subgraph, imply that necessarily $(C, \ell) = (C', \ell')$, and therefore $(G_{H,\varphi},\sigma)$ satisfies property~{\sf P3} and the theorem follows by Lemma~\ref{lem:properties}.
\end{proof}

When $H$ is a connected graph, the lower bound of Theorem~\ref{thm:lower-bound-colors}
together with the algorithms given by Proposition~\ref{Kh_single} and Theorem~\ref{thm:generic-algo-colors} completely settle, under the \ETH, the asymptotic complexity of \textsc{Colorful $H$-IS-Deletion} parameterized by treewidth. Note that, in particular, Theorem~\ref{thm:lower-bound-colors} applies when $H$ is path, in contrast to the subgraph version that can be solved in polynomial time~\cite{CyganMPP17}.

Therefore, what remains is to obtain tight lower bounds when $H$ is disconnected. In particular, Theorem~\ref{thm:lower-bound-colors} cannot be applied at all when all the connected components of $H$ are cliques, since the machinery that we developed (inspired by Cygan et al.~\cite{CyganMPP17}) using the framework graph $F_{H,\varphi}$ crucially needs two non-adjacent vertices in the same connected component. Let us now focus on those graphs, sometimes called \emph{cluster} graphs in the literature.

%As for any graph $H$ and any instance $(G,\sigma)$ of \textsc{Colorful $H$-IS-Deletion}, any edge between two vertices $u,v$ with $\sigma(u)=\sigma(v)$ can be safely deleted without affecting the instance, the \textsc{Colorful $K_2$-IS-Deletion} problem is equivalent to computing a minimum vertex cover in a bipartite graph, which can be done in polynomial time. Similarly, the \textsc{Colorful $I_2$-IS-Deletion} problem can also be solved in polynomial time, by computing a minimum vertex cover in the bipartite complement of the input graph.

As mentioned in Section~\ref{sec:algos-no-colors}, both \textsc{Colorful $K_2$-IS-Deletion} and \textsc{Colorful $I_2$-IS-Deletion} can be solved in polynomial time. In our next result we show that if $H$ is slightly larger than these two graphs (namely, $K_2$ or $I_2$), then \textsc{Colorful $H$-IS-Deletion} becomes hard. Namely, we provide a single-exponential lower bound for the following three graphs $H$ on three vertices that are {\sl not} covered by Theorem~\ref{thm:lower-bound-colors}: $K_3$, $I_3$, and $K_2 + K_1$. Note that these lower bounds are tight by the algorithm of Theorem~\ref{thm:generic-algo-colors}.

\begin{theorem}\label{thm:LBI3}
Let  $H \in \{K_{3}, I_3, K_2 + K_1\}$. Then, unless the \ETH fails, the  {\sc Colorful $H$-IS-Deletion} problem cannot be solved in time $\Ocal^*(2^{o(t)})$, where $t$ is the width of a given tree decomposition of the input graph.
\end{theorem}
\begin{proof}
We will prove that none of the considered problems can be solved in time $2^{o(n)}$ under the \ETH, which clearly implies the statement of the theorem. For this, we reduce from the \textsc{Vertex Cover} problem restricted to input graphs with maximum degree at most three. To see that this problem cannot be solved in time $2^{o(n)}$ under the \ETH, where $n$ is the number of vertices of the input graph, one can apply the classical \NP-hardness reduction~\cite{GareyJ79} from 3-\textsc{Sat} to \textsc{Vertex Cover}, but restricting the input formulas to be {\sl clean}. Then the result follows from Lemma~\ref{lem:clean3SAT}.

We first present a reduction for $H=K_3$. Given an instance $G$ of \textsc{Vertex Cover}, with $|V(G)|=n$, $|E(G)|=m$, and $\Delta(G) \leq 3$, we proceed to construct an instance $(G_{K_3}, \sigma)$ of \textsc{$H$-IS-Deletion} with $|V(G_{K_3})| = \Ocal(n)$ such that $G$ has a vertex cover of size at most $k$ if and only if $G_{K_3}$ has a set of size at most $k+m$ hitting all (induced) $\sigma$-$K_3$-subgraphs. Note that this will prove the desired result, as $\tw(G_{K_3}) \leq  |V(G_{K_3})| = \Ocal(n)$, and a tree decomposition of $G_{K_3}$ achieving that width consists of just one bag containing all vertices.

Let $V(K_3) = \{z_1,z_2,z_3\}$ and let $L$ be the graph obtained from three disjoint copies of $K_3$ by identifying vertices $z_2$ of the first and second copies, and vertices $z_3$ of the second and third copies. We define $G_{K_3}$ as the graph obtained from $G$ be replacing each edge $\{u,v\} \in E(G)$ by the graph $L$, identifying vertex $u$ (resp. $v$) with vertex $z_1$ of the first (resp. third) copy of $K_3$ in $L$. We define the $K_3$-coloring $\sigma$ of $G_{K_3}$ in the natural way, that is, each vertex of $G_{K_3}$ gets the color of its corresponding vertex in the gadget $L$. Note that all vertices that were originally in $G$ get color $z_1$. Since $\Delta(G) \leq 3$, it follows that $|V(G_{K_3})|  = |V(G)| + |E(G)|\cdot (|V(L)| -2 ) = |V(G)| + 5|E(G)| \leq 17 |V(G)|/2 = \Ocal(n)$.

By construction of $G_{K_3}$, each edge of $G$ gives rise to exactly three (induced) $\sigma$-$K_3$-subgraphs in $G_{K_3}$, and all $\sigma$-$K_3$-subgraphs in $G_{K_3}$ are of this type. For each gadget $L$, there are exactly two vertex sets of minimum size hitting its three $\sigma$-$K_3$-subgraphs, of size two, each of them containing exactly one of the original vertices of $G$. Therefore, a vertex cover of $G$ of size at most $k$ can be easily transformed into a set $X \subseteq V(G_{K_3})$ of size at most $k+m$ hitting all $\sigma$-$K_3$-subgraphs of $G_{K_3}$, and vice versa.

\smallskip

Let now $H=I_3$, with $V(I_3) = \{z_1,z_2,z_3\}$. Given an instance $G$ of \textsc{Vertex Cover}, with $|V(G)|=n$ and $\Delta(G) \leq 3$,  we start with the instance $(G_{K_3},\sigma)$ of \textsc{$K_3$-IS-Deletion} defined above, and we construct an instance $(G_{I_3},\sigma')$ of \textsc{$I_3$-IS-Deletion} such that $V(G_{I_3}) = V(G_{K_3})$, $\sigma' = \sigma$ (by associating the labels of $V(K_3)$ and $V(I_3)$), and $E(G_{I_3})$ defined as the tripartite complement of $E(G_{K_3})$, that is, for every pair of vertices $u,v \in V(G_{I_3})$, $\{u,v\} \in E(G_{I_3})$ if and only if $\sigma'(u) \neq \sigma'(v)$ and $\{u,v\} \notin E(G_{K_3})$. Since $|V(G_{I_3})| = |V(G_{K_3})| = \Ocal(n)$ and there is a one-to-one correspondence between induced $\sigma$-$K_3$-subgraphs in $G_{K_3}$ and induced $\sigma$-$I_3$-subgraphs in $G_{I_3}$, the result follows.

\smallskip

Finally, let now $H=K_2 + K_1$, with $V(H) = \{z_1,z_2,z_3\}$ such that $z_1$ and $z_2$ are adjacent. Similarly, we construct an instance $(G_{K_2+K_1},\sigma)$ of \textsc{$(K_2+K_1)$-IS-Deletion} starting from $(G_{K_3},\sigma)$, but in this case we only complement the neighborhood of the vertices $u \in V(G_{K_2+K_1})$ with $\sigma(u) = z_3$, keeping the set of vertices colored $z_3$ an independent set. Again, there is a one-to-one correspondence between induced $\sigma$-$K_3$-subgraphs in $G_{K_3}$ and induced $\sigma$-$(K_2 + K_1)$-subgraphs in $G_{K_2 + K_1}$, and the proof is complete.
 \end{proof}

The proof of Theorem~\ref{thm:LBI3} can be easily adapted to $H=P_3$ by complementing the appropriate neighborhoods, hence obtaining a lower bound of $\Ocal^*(2^{o(t)})$ for
\textsc{Colorful $P_3$-IS-Deletion}. Note, however, that this lower for $P_3$ bound already follows from Theorem~\ref{thm:lower-bound-colors}.%, since $P_3$ fits the conditions of the theorem.

It is also easy to adapt the proof of Theorem~\ref{thm:LBI3} to larger graphs, but then the lower bound of $\Ocal^*(2^{o(t)})$ is not tight anymore. For example, for $H=2K_2$ (the disjoint union of two edges), with $V(H) = \{z_1,z_2,z_3,z_4\}$ such that the edges are $\{z_1,z_2\}$ and $\{z_3,z_4\}$, it suffices to take the instance $(G_{K_2+K_1},\sigma)$ of \textsc{$(K_2+K_1)$-IS-Deletion} defined above and to add a private neighbor colored $z_4$ for every vertex of $G_{K_2+K_1}$ colored $z_3$. Also, for $H=K_h$ with $h \geq 4$, in the gadget $L$ we just replace the triangles by cliques of size $h$, and for $H=I_h$ with $h \geq 4$, we  take the $h$-partite complement of the previous instance of \textsc{$K_h$-IS-Deletion}.

%\ig{Another result:} there is no subexponential algorithm for $K_2+K_1$ (reduction from \textsc{Vertex Cover} in subcubic graphs). The reduction can be generalized to other small graphs (see picture).

%\begin{theorem}
%Let $H$ be a disconnected graph with some component that has at least three vertices and is not a clique. Then \textsc{Colorful $H$-IS-Deletion} cannot be solved in time $2^{o(t^{h-2})} \cdot n^{\Ocal(1)}$, unless the \ETH fails.
%\end{theorem}

%\ig{Another result:} there is no subexponential algorithm for $K_2+K_1$ (reduction from \textsc{Vertex Cover} in subcubic graphs). The reduction can be generalized to other small graphs (see picture).

\section{Further research}
\label{sec:concl}

Concerning \textsc{$H$-IS-Deletion}, the complexity gap is still quite large for most graphs $H$, as our lower bounds (Theorems~\ref{thm:lower-bound-no-colors},~\ref{thm:lower-bound-no-colors-WEAKER},~\ref{thm:lower-bound-no-colors-2}, and~\ref{thm:lower-bound-complete-bipartite})
only apply to graphs $H$ that are ``close'' to cliques or complete bipartite graphs. In particular, Theorem~\ref{thm:lower-bound-no-colors} provides tight bounds for $P_3$ or $K_4-e$ (the {\sl diamond}), but we do not know the tight function $f_H(t)$ for other small graphs $H$ on four vertices such as $P_4$, $K_{1,3}$ (the {\sl claw}), or $2K_2$.

We think that for most graphs $H$ on $h$ vertices, the upper bound
$f_H(t) = 2^{\Ocal(t^{h-2})}$ given by Theorem~\ref{thm:generic-algo-no-colors} is the asymptotically tight function, and that the single-exponential algorithms for cliques and independent sets are isolated exceptions. The reason is that, in contrast to the subgraph version, when hitting induced subgraphs, edges and non-edges behave essentially in the same way when performing dynamic programming, as one has to keep track of both the existence and the non-existence of edges in order to construct the tables, and storing this information seems to be unavoidable.

Concerning the algorithm for \textsc{$I_h$-IS-Deletion} running in time  $2^{\Ocal(t)} \cdot n^{h}$ (Theorem~\ref{thm:algo-Is}), %whose term $2^{\Ocal(t)}$ is tight for $I_3$ by Theorem~\ref{thm:LBI3}, 
it would be interesting to find, if it exists, an \FPT algorithm parameterized by both $t$ and $h$ while keeping the dependency on $t$ single-exponential, maybe even being linear in $n$.

\smallskip

As for \textsc{Colorful $H$-IS-Deletion}, in view of Theorems~\ref{thm:generic-algo-colors},~\ref{Kh_single},~\ref{thm:lower-bound-colors}, and~\ref{thm:LBI3}, only the cases where $H$ is a disjoint union of at least two cliques and $|V(H)| \geq 4$ remain open. In particular, we do not know the tight function when $H$ is an independent set or a matching with $|V(H)| \geq 4$.

%\ig{Say that if $H$ is a path, there is a strong difference with the non-induced version}

%\ig{Mention differences between induced/non-induced -- say which of our reductions apply, and which not -- are they optimal according to Cygan et al.~\cite{CyganMPP17}?}

%\ig{$P_3 + I_1$: compare with single-exponential algorithm for the subgraph version, using the algorithm hitting paths as minors~\cite{BasteST20-part2} -- I don't think we have a reduction}

\medskip

\noindent \textbf{Acknowledgement}. We would like to thank the anonymous reviewers for helpful remarks that improved the presentation of the manuscript, in particular for pointing out a previous flaw in the proof of Theorem~\ref{thm:lower-bound-no-colors-WEAKER}.

%\newpage
\bibliographystyle{abbrv}
\bibliography{Biblio}

\begin{thebibliography}{10}

\bibitem{AdlerDFST11}
Isolde Adler, Frederic Dorn, Fedor~V. Fomin, Ignasi Sau, and Dimitrios~M.
  Thilikos.
\newblock Faster parameterized algorithms for minor containment.
\newblock {\em Theoretical Computer Science}, 412(50):7018--7028, 2011.
\newblock \href {https://doi.org/10.1016/j.tcs.2011.09.015}
  {\path{doi:10.1016/j.tcs.2011.09.015}}.

\bibitem{BasteST20}
Julien Baste, Ignasi Sau, and Dimitrios~M. Thilikos.
\newblock A complexity dichotomy for hitting connected minors on bounded
  treewidth graphs: the chair and the banner draw the boundary.
\newblock In {\em Proc. of the 31st Annual ACM-SIAM Symposium on Discrete
  Algorithms (SODA)}, pages 951--970, 2020.
\newblock \href {https://doi.org/10.1137/1.9781611975994.57}
  {\path{doi:10.1137/1.9781611975994.57}}.

\bibitem{BasteST20-part1}
Julien Baste, Ignasi Sau, and Dimitrios~M. Thilikos.
\newblock Hitting minors on bounded treewidth graphs. {I.} {G}eneral upper
  bounds.
\newblock {\em {SIAM} Journal on Discrete Mathematics}, 34(3):1623--1648, 2020.
\newblock \href {https://doi.org/10.1137/19M1287146}
  {\path{doi:10.1137/19M1287146}}.

\bibitem{BasteST20-part2}
Julien Baste, Ignasi Sau, and Dimitrios~M. Thilikos.
\newblock {Hitting minors on bounded treewidth graphs. {II.} Single-exponential
  algorithms}.
\newblock {\em Theoretical Computer Science}, 814:135--152, 2020.
\newblock \href {https://doi.org/10.1016/j.tcs.2020.01.026}
  {\path{doi:10.1016/j.tcs.2020.01.026}}.

\bibitem{BasteST20-part3}
Julien Baste, Ignasi Sau, and Dimitrios~M. Thilikos.
\newblock {Hitting minors on bounded treewidth graphs. {III.} Lower bounds}.
\newblock {\em Journal of Computer and System Sciences}, 109:56--77, 2020.
\newblock \href {https://doi.org/10.1016/j.jcss.2019.11.002}
  {\path{doi:10.1016/j.jcss.2019.11.002}}.

\bibitem{BodlaenderCKN15}
Hans~L. Bodlaender, Marek Cygan, Stefan Kratsch, and Jesper Nederlof.
\newblock Deterministic single exponential time algorithms for connectivity
  problems parameterized by treewidth.
\newblock {\em Information and Computation}, 243:86--111, 2015.
\newblock \href {https://doi.org/10.1016/j.ic.2014.12.008}
  {\path{doi:10.1016/j.ic.2014.12.008}}.

\bibitem{BodlaenderDDFLP16}
Hans~L. Bodlaender, P{\aa}l~Gr{\o}n{\aa}s Drange, Markus~S. Dregi, Fedor~V.
  Fomin, Daniel Lokshtanov, and Michal Pilipczuk.
\newblock {A $c^k n$ $5$-Approximation Algorithm for Treewidth}.
\newblock {\em {SIAM} Journal on Computing}, 45(2):317--378, 2016.
\newblock \href {https://doi.org/10.1137/130947374}
  {\path{doi:10.1137/130947374}}.

\bibitem{BodlaenderHL14grap}
Hans~L. Bodlaender, Pinar Heggernes, and Daniel Lokshtanov.
\newblock Graph modification problems (dagstuhl seminar 14071).
\newblock {\em Dagstuhl Reports}, 4(2):38--59, 2014.
\newblock \href {https://doi.org/10.4230/DagRep.4.2.38}
  {\path{doi:10.4230/DagRep.4.2.38}}.

\bibitem{claw-free}
Flavia Bonomo{-}Braberman, Julliano~R. Nascimento, Fabiano de~S.~Oliveira,
  U{\'{e}}verton~S. Souza, and Jayme~Luiz Szwarcfiter.
\newblock Linear-time algorithms for eliminating claws in graphs.
\newblock In {\em Proc. of the 26th International Conference on Computing and
  Combinatorics (COCOON)}, volume 12273 of {\em LNCS}, pages 14--26, 2020.
\newblock \href {https://doi.org/10.1007/978-3-030-58150-3\_2}
  {\path{doi:10.1007/978-3-030-58150-3\_2}}.

\bibitem{Courcelle90}
Bruno Courcelle.
\newblock {The Monadic Second-Order Logic of Graphs. I. Recognizable Sets of
  Finite Graphs}.
\newblock {\em Information and Computation}, 85(1):12--75, 1990.
\newblock \href {https://doi.org/10.1016/0890-5401(90)90043-H}
  {\path{doi:10.1016/0890-5401(90)90043-H}}.

\bibitem{CrespelleDFG13asur}
Christophe Crespelle, Pål Grønås, Drange, Fedor~V. Fomin, and Petr~A.
  Golovach.
\newblock A survey of parameterized algorithms and the complexity of edge
  modification.
\newblock {\em CoRR}, abs/2001.06867, 2013.
\newblock URL: \url{https://arxiv.org/abs/2001.06867}.

\bibitem{CyganFKLMPPS15}
Marek Cygan, Fedor~V. Fomin, Lukasz Kowalik, Daniel Lokshtanov, D{\'{a}}niel
  Marx, Marcin Pilipczuk, Michal Pilipczuk, and Saket Saurabh.
\newblock {\em Parameterized Algorithms}.
\newblock Springer, 2015.
\newblock \href {https://doi.org/10.1007/978-3-319-21275-3}
  {\path{doi:10.1007/978-3-319-21275-3}}.

\bibitem{CyganMPP17}
Marek Cygan, D{\'{a}}niel Marx, Marcin Pilipczuk, and Michal Pilipczuk.
\newblock Hitting forbidden subgraphs in graphs of bounded treewidth.
\newblock {\em Information and Computation}, 256:62--82, 2017.
\newblock \href {https://doi.org/10.1016/j.ic.2017.04.009}
  {\path{doi:10.1016/j.ic.2017.04.009}}.

\bibitem{CyganNPPRW11}
Marek Cygan, Jesper Nederlof, Marcin Pilipczuk, Michal Pilipczuk, Johan M.~M.
  van Rooij, and Jakub~Onufry Wojtaszczyk.
\newblock Solving connectivity problems parameterized by treewidth in single
  exponential time.
\newblock In {\em Proc. of the 52nd Annual {IEEE} Symposium on Foundations of
  Computer Science (FOCS)}, pages 150--159, 2011.
\newblock \href {https://doi.org/10.1109/FOCS.2011.23}
  {\path{doi:10.1109/FOCS.2011.23}}.

\bibitem{Diestel12}
Reinhard Diestel.
\newblock {\em Graph Theory, 4th Edition}, volume 173 of {\em Graduate texts in
  mathematics}.
\newblock Springer, 2012.
\newblock URL: \url{https://dblp.org/rec/books/daglib/0030488.bib}.

\bibitem{DornPBF10}
Frederic Dorn, Eelko Penninkx, Hans~L. Bodlaender, and Fedor~V. Fomin.
\newblock {Efficient Exact Algorithms on Planar Graphs: Exploiting Sphere Cut
  Decompositions}.
\newblock {\em Algorithmica}, 58(3):790--810, 2010.
\newblock \href {https://doi.org/10.1007/s00453-009-9296-1}
  {\path{doi:10.1007/s00453-009-9296-1}}.

\bibitem{DF13}
Rodney~G. Downey and Michael~R. Fellows.
\newblock {\em Fundamentals of Parameterized Complexity}.
\newblock Texts in Computer Science. Springer, 2013.
\newblock \href {https://doi.org/10.1007/978-1-4471-5559-1}
  {\path{doi:10.1007/978-1-4471-5559-1}}.

\bibitem{FominLPS16}
Fedor~V. Fomin, Daniel Lokshtanov, Fahad Panolan, and Saket Saurabh.
\newblock Efficient computation of representative families with applications in
  parameterized and exact algorithms.
\newblock {\em Journal of the {ACM}}, 63(4):29:1--29:60, 2016.
\newblock \href {https://doi.org/10.1145/2886094} {\path{doi:10.1145/2886094}}.

\bibitem{FominSM15grap}
Fedor~V. Fomin, Saket Saurabh, and Neeldhara Misra.
\newblock Graph modification problems: {A} modern perspective.
\newblock In {\em Proc. of the 9th International Frontiers in Algorithmics
  Workshop (FAW)}, volume 9130 of {\em LNCS}, pages 3--6, 2015.
\newblock \href {https://doi.org/10.1007/978-3-319-19647-3\_1}
  {\path{doi:10.1007/978-3-319-19647-3\_1}}.

\bibitem{GareyJ79}
Michael~R. Garey and David~S. Johnson.
\newblock {\em {Computers and Intractability: {A} Guide to the Theory of
  NP-Completeness}}.
\newblock W. H. Freeman, 1979.
\newblock URL: \url{https://dblp.org/rec/books/fm/GareyJ79.bib}.

\bibitem{ImpagliazzoP01-ETH}
Russell Impagliazzo and Ramamohan Paturi.
\newblock {On the Complexity of $k$-SAT}.
\newblock {\em Journal of Computer and System Sciences}, 62(2):367--375, 2001.
\newblock \href {https://doi.org/10.1006/jcss.2000.1727}
  {\path{doi:10.1006/jcss.2000.1727}}.

\bibitem{ImpagliazzoP01}
Russell Impagliazzo, Ramamohan Paturi, and Francis Zane.
\newblock {Which Problems Have Strongly Exponential Complexity?}
\newblock {\em Journal of Computer and System Sciences}, 63(4):512--530, 2001.
\newblock \href {https://doi.org/10.1006/jcss.2001.1774}
  {\path{doi:10.1006/jcss.2001.1774}}.

\bibitem{JansenLS14}
Bart M.~P. Jansen, Daniel Lokshtanov, and Saket Saurabh.
\newblock {A Near-Optimal Planarization Algorithm}.
\newblock In {\em Proc. of the 25th Annual {ACM-SIAM} Symposium on Discrete
  Algorithms (SODA)}, pages 1802--1811, 2014.
\newblock \href {https://doi.org/10.1137/1.9781611973402.130}
  {\path{doi:10.1137/1.9781611973402.130}}.

\bibitem{Klo94}
Ton Kloks.
\newblock {\em Treewidth. Computations and Approximations}.
\newblock Springer-Verlag LNCS, 1994.
\newblock \href {https://doi.org/10.1007/BFb0045375}
  {\path{doi:10.1007/BFb0045375}}.

\bibitem{LeYa80}
John~M. Lewis and Mihalis Yannakakis.
\newblock {The Node-Deletion Problem for Hereditary Properties is NP-Complete}.
\newblock {\em Journal of Computer and System Sciences}, 20(2):219--230, 1980.
\newblock \href {https://doi.org/10.1016/0022-0000(80)90060-4}
  {\path{doi:10.1016/0022-0000(80)90060-4}}.

\bibitem{LokshtanovMS11}
Daniel Lokshtanov, D{\'{a}}niel Marx, and Saket Saurabh.
\newblock {Lower bounds based on the Exponential Time Hypothesis}.
\newblock {\em Bulletin of the {EATCS}}, 105:41--72, 2011.

\bibitem{Pilipczuk17}
Marcin Pilipczuk.
\newblock {A tight lower bound for Vertex Planarization on graphs of bounded
  treewidth}.
\newblock {\em Discrete Applied Mathematics}, 231:211--216, 2017.
\newblock \href {https://doi.org/10.1016/j.dam.2016.05.019}
  {\path{doi:10.1016/j.dam.2016.05.019}}.

\bibitem{Pilipczuk11}
Michal Pilipczuk.
\newblock Problems parameterized by treewidth tractable in single exponential
  time: {A} logical approach.
\newblock In {\em Proc. of the 36th International Symposium on Mathematical
  Foundations of Computer Science (MFCS)}, volume 6907 of {\em LNCS}, pages
  520--531, 2011.
\newblock \href {https://doi.org/10.1007/978-3-642-22993-0\_47}
  {\path{doi:10.1007/978-3-642-22993-0\_47}}.

\bibitem{RueST14}
Juanjo Ru{\'{e}}, Ignasi Sau, and Dimitrios~M. Thilikos.
\newblock Dynamic programming for graphs on surfaces.
\newblock {\em {ACM} Transactions on Algorithms}, 10(2):8:1--8:26, 2014.
\newblock \href {https://doi.org/10.1145/2556952} {\path{doi:10.1145/2556952}}.

\end{thebibliography}

 \clearpage
 \appendix

\end{document}